\documentclass[twocolumn,reprint,pra,aps,superscriptaddress,nofootinbib,9pt]{revtex4-1}  
\usepackage{epsf,graphicx}
\usepackage{amsmath,amssymb,amsfonts,array}
\usepackage{multirow}

\usepackage{units}

\usepackage{caption}
\usepackage{subfig}
\usepackage{hhline}
\usepackage{textcomp}
\usepackage{bm} % fonts 
\usepackage{float} % to stop tables floating here and there
\usepackage{rotating} % for sidewaytable
\usepackage{amsthm}
\usepackage{esvect} % for vectors: use vv{...}
\usepackage{color}
\usepackage[dvipsnames]{xcolor}

\usepackage{enumitem}	% to control space between lines for enumerated text

\usepackage{xparse}
\ExplSyntaxOn
\NewDocumentCommand{\mref}{m}{\quinn_mref:n {#1}}
\seq_new:N \l_quinn_mref_seq
\cs_new:Npn \quinn_mref:n #1
{
	\seq_set_split:Nnn \l_quinn_mref_seq { , } { #1 }
	\seq_pop_right:NN \l_quinn_mref_seq \l_tmpa_tl
	( % print the left parenthesis
	\seq_map_inline:Nn \l_quinn_mref_seq
	{ \ref{##1}, } % print the first references
	\exp_args:NV \ref \l_tmpa_tl % print the last or only one
	) % print the right parenthesis
}
\ExplSyntaxOff
\usepackage{tikz}
\usetikzlibrary{decorations.pathmorphing,patterns}
\usetikzlibrary{matrix} % for block alignment
\usetikzlibrary{arrows} % for arrow heads
\usetikzlibrary{calc} % for manimulation of coordinates
\usepackage[
bookmarks=true,
colorlinks=true,
linkcolor=blue,
urlcolor=blue,
citecolor=blue,
pdftex,
bookmarks=true,
linktocpage=true, % makes the page number as hyperlink in table of content
hyperindex=true
]{hyperref}
\usepackage{bbm,bbold} %-- for \mathbbm{1}
\usepackage[mathscr]{euscript}
\DeclareSymbolFont{rsfs}{U}{rsfs}{m}{n}
\DeclareSymbolFontAlphabet{\mathscrsfs}{rsfs}
\usepackage[normalem]{ulem}

\makeatletter
\renewcommand*{\@seccntformat}[1]{\csname the#1\endcsname.\hspace{.5em}}
\makeatother

\theoremstyle{plain}
\newtheorem{theorem}{Theorem}
\newtheorem{proposition}{Proposition} 
 
\newtheorem{note}{Note}

\newtheorem{lemma}{Lemma}  
\newtheorem{example}{Example}
\newtheorem{remark}{Remark}

\newcommand*{\Ovec}{\mathbf{0}} % for zero vec
\newcommand*{\Dp}{\mathscrsfs{D}} % for var \mathcal{D}{+}
\newcommand*{\Dm}{ \overline{\mathscrsfs{D}} } % for var \mathcal{D}{-}
\newcommand*{\U}{\mathcal{U}} % for unitary operator
\newcommand*{\A}{\mathcal{A}} % r.v. for Alice
\newcommand*{\B}{\mathcal{B}} % r.v. for Bob
\newcommand*{\E}{\mathcal{E}} % r.v. for Eve's measurement outcome
\newcommand*{\EE}{\mathbb{E}} % for operators
\newcommand*{\tE}[1]{\ket{\tilde{E_{#1}}}} % for operators

\newcommand*{\fB}{\mathfrak{B}}

\newcommand*{\etal}{\textit{et al}}
\mathchardef\mhyphen="2D
\newcommand{\mh}{\mhyphen}
\newcommand{\xy}{x\mh y}
\newcommand{\uv}{u\mh v}
\newcommand{\Span}{\text{span}}
\newcommand{\sgn}{\text{sgn}}
\newcommand*{\tr}{\text{tr}}
\newcommand*{\Tr}{\text{Tr}}
\DeclareMathOperator*{\argmax}{argmax} % no space, limits underneath in displays
\newcommand*{\wt}{\widetilde}
\newcommand*{\Trpart}[1]{\text{Tr}_{\text{#1}}}
\newcommand{\tab}{~~~~~~}
\newcommand{\half}{\frac{1}{2}} 
\newcommand{\invsqrttwo}{\frac{1}{\sqrt{2}}}
\renewcommand{\l}{\lambda}
\renewcommand{\a}{\alpha}
\renewcommand{\b}{\beta}
\newcommand{\m}{\mu}
\newcommand{\n}{\nu}

\newcommand{\El}{E_{\l}}

\newcommand{\tensorprod}{\otimes}

\newcommand{\ket}[1]{| #1 \rangle}
\newcommand{\braket}[2]{\langle #1 | #2 \rangle}
\newcommand{\ketbra}[2]{| #1 \rangle\langle #2 |}
\newcommand{\braopket}[3]{\langle #1 | #2 | #3 \rangle}
\newcommand{\Xlambdax}[2]{B_{#2}\tensorprod\sqrt{\El}~ \ket{#1}} 

\newcommand*{\concave}[1]{\sqrt{{#1}\left(1 - {#1} \right)}}
\newcommand*{\concavesq}[1]{\sqrt{{#1}\left[1 - {#1} \right]}}
\newcommand*{\phiconcave}[1]{\phi{\left[2\concave{{#1}}\right]}} 
\newcommand*{\philog}[1]{\left(1+{#1}\right) \ln \left(1+{#1}\right) + \left(1-{#1}\right) \ln \left(1-{#1}\right)}

\newcommand*{\factor}[1]{\sqrt{\frac{D_{uv}}{1-D_{uv}}}}
\newcommand{\wh}{\widehat}
\newcommand{\xuvExpr}[2]{\invsqrttwo \left(\ket{#1}+\ket{#2}\right)}
\newcommand{\yuvExpr}[2]{\invsqrttwo \left(\ket{#1}-\ket{#2}\right)}
\newcommand*{\entangledExpr}[5]{\sqrt{1-D_{#1}}~\ket{#2}\ket{\xi_{#3}} + \sqrt{D_{#1}}~\ket{#4}\ket{\zeta_{#5}} }
\newcommand*{\entangledExprGen}[5]{\sqrt{#1}~\ket{#2}\ket{\xi_{#3}} + \sqrt{1-{#1}}~\ket{#4}\ket{\zeta_{#5}} }

\newcommand*{\txx}{\ket{x}\ket{x}}
\newcommand*{\tyy}{\ket{y}\ket{y}}
\newcommand*{\txy}{\ket{x}\ket{y}}
\newcommand*{\tyx}{\ket{y}\ket{x}}
\newcommand*{\Phixy}[1]{\ket{\Phi_{xy}^{#1}}}
\newcommand*{\PhixyExpr}[1]{\invsqrttwo\left(\txx{#1}\tyy\right)}
\newcommand*{\Psixy}[1]{\ket{\Psi_{xy}^{#1}}}
\newcommand*{\PsixyExpr}[1]{\invsqrttwo\left(\txy{#1}\tyx\right)}

\newcommand*{\xyscaled}[1]{\left(\cos{#1}\ket{x} + \sin{#1}\ket{y}\right)}  
\newcommand*{\rot}{\mathbf{R}}
\newcommand*{\vegn}{\mathbf{e}}
\newcommand*{\vcan}{\boldsymbol{\varepsilon}}
\newcommand*{\matd}{\mathbf{D}}

\def\spvecA#1;{\if;#1;\else #1\cr \expandafter \spvecA \fi}
\newcommand*{\intAA}{
	\Dp_{uv}~\ket{\E_0} + \Dm_{uv}~\ket{\E_3}
}
\newcommand*{\intAB}{
	\Dm_{uv}~\ket{\E_0} + \Dp_{uv}~\ket{\E_3}
}
\newcommand*{\intAC}{
	\Dm_{uv}~\ket{\E_2} + \Dp_{uv}~\ket{\E_1}
}
\newcommand*{\intAD}{
	\Dp_{uv}~\ket{\E_2} + \Dm_{uv}~\ket{\E_1}
}
\newcommand*{\povmAA}{
	\ket{\E_0}
}
\newcommand*{\povmAB}{
	\ket{\E_1}
}
\newcommand*{\povmAC}{
	\ket{\E_2}
}
\newcommand*{\povmAD}{
	\ket{\E_3}
}
\newcommand*{\intBA}{
	\ket{\E_0}
}
\newcommand*{\intBB}{
	2\Dp\Dm~\ket{\E_0}+\left(\Dp^2-\Dm^2\right)\ket{\E_1}
}
\newcommand*{\intBC}{
	\ket{\E_2}
}
\newcommand*{\intBD}{
	2\Dp\Dm~\ket{\E_2}+\left(\Dp^2-\Dm^2\right)\ket{\E_3}
}
\newcommand*{\povmBA}{
	\Dp\ket{\E_0}-\Dm\ket{\E_1}
}
\newcommand*{\povmBB}{
	\Dm\ket{\E_0}+\Dp\ket{\E_1}
}
\newcommand*{\povmBC}{
	\Dp\ket{\E_2}-\Dm\ket{\E_3}
}
\newcommand*{\povmBD}{
	\Dm\ket{\E_2}+\Dp\ket{\E_3}
}
\newcommand*{\intCA}{
	\sqrt{1-D_{uv}}~\ket{\E_0} - \sqrt{D_{uv}}~\ket{\E_1}
}
\newcommand*{\intCB}{
	\sqrt{1-D_{uv}}~\ket{\E_0} + \sqrt{D_{uv}}~\ket{\E_1}
}
\newcommand*{\intCC}{
	\sqrt{1-D_{uv}}~\ket{\E_2} - \sqrt{D_{uv}}~\ket{\E_3}
}
\newcommand*{\intCD}{
	\sqrt{1-D_{uv}}~\ket{\E_2} + \sqrt{D_{uv}}~\ket{\E_3}
}
\newcommand*{\povmCA}{
	\invsqrttwo\left(\ket{\E_0}-\ket{\E_1}\right)
}
\newcommand*{\povmCB}{
	\invsqrttwo\left(\ket{\E_0}+\ket{\E_1}\right)
}
\newcommand*{\povmCC}{
	\invsqrttwo\left(\ket{\E_2}-\ket{\E_3}\right)
}
\newcommand*{\povmCD}{
	\invsqrttwo\left(\ket{\E_2}+\ket{\E_3}\right)
}
\newcommand*{\intDA}{
	\Dp_{uv}~\ket{E_0} + \Dm_{uv}~\ket{E_1}
}
\newcommand*{\intDB}{
	\Dm_{uv}~\ket{E_0} + \Dp_{uv}~\ket{E_1}
}
\newcommand*{\intDC}{
	\Dp_{uv}~\ket{E_2} + \Dm_{uv}~\ket{E_3}
}
\newcommand*{\intDD}{
	\Dm_{uv}~\ket{E_2} + \Dp_{uv}~\ket{E_3}
}
\newcommand*{\intEA}{
	\left(\Dp_{uv}\sqrt{a}+\Dm_{uv}\sqrt{1-a}\right)\ket{\E_0} + \left(\Dm_{uv}\sqrt{a}-\Dp_{uv}\sqrt{1-a}\right)\ket{\E_1}
}
\newcommand*{\intEB}{
	\left(\Dm_{uv}\sqrt{a}+\Dp_{uv}\sqrt{1-a}\right)\ket{\E_0} + \left(\Dp_{uv}\sqrt{a}-\Dm_{uv}\sqrt{1-a}\right)\ket{\E_1}
}
\newcommand*{\intEC}{
	\left(\Dp_{uv}\sqrt{a}+\Dm_{uv}\sqrt{1-a}\right)\ket{\E_2} + \left(\Dm_{uv}\sqrt{a}-\Dp_{uv}\sqrt{1-a}\right)\ket{\E_3}
}
\newcommand*{\intED}{
	\left(\Dm_{uv}\sqrt{a}+\Dp_{uv}\sqrt{1-a}\right)\ket{\E_2} + \left(\Dp_{uv}\sqrt{a}-\Dm_{uv}\sqrt{1-a}\right)\ket{\E_3}
}
\newcommand*{\povmEA}{
	\sqrt{a}\ket{\E_0}-\sqrt{1-a}\ket{\E_1}
}
\newcommand*{\povmEB}{
	\sqrt{1-a}\ket{\E_0}+\sqrt{a}\ket{\E_1}
}
\newcommand*{\povmEC}{
	\sqrt{a}\ket{\E_2}-\sqrt{1-a}\ket{\E_3}
}
\newcommand*{\povmED}{
	\sqrt{1-a}\ket{\E_2}+\sqrt{a}\ket{\E_3}
}
\tikzstyle{startstop} = [rectangle, rounded corners, minimum width=1cm, minimum height=.5cm,text centered, draw=black, fill=cyan!30]
\tikzstyle{startstop2} = [rectangle, rounded corners, minimum width=1cm, minimum height=.5cm,text centered, text width = 2 cm, draw=black, fill=yellow!30] %<text width fixed>
\tikzstyle{startstopflex} = [rectangle, rounded corners, minimum width=1cm, minimum height=.5cm,text centered, draw=black, fill=cyan!30] %<text width varied>
\tikzstyle{temp} = [rectangle, minimum width=1cm, minimum height=.5cm,text centered, draw=black, fill=white!30] %<text width fixed>
\tikzstyle{arrow} = [thick,->,>=stealth]
\tikzstyle{arrowDashed} = [dashed,->,>=stealth]
\tikzstyle{line} = [thick]
\tikzstyle{lineDashed} = [dashed,thin]
\tikzstyle{branch} = [circle,inner sep=0pt,minimum size=1mm,fill=black,draw=black]
\tikzstyle{block} = [draw,rectangle,thin,minimum height=2em,minimum width=2em, fill=magenta!10]
\tikzstyle{tinyblock} = [draw,rectangle,thin,minimum height=2em,minimum width=2em]
\tikzstyle{sum} = [draw, fill=blue!20, circle, node distance=1cm]

\newcommand*{\tikzGist}{     
	\begin{tikzpicture}[node distance=2cm] 
	
	\node (intr) [startstop] {Optimal Interaction};	
	
	\node (many) [startstopflex, xshift=-6cm, yshift=-2cm] {
		\begin{tabular}{c}
		infinitely many: \\ 
		$\{\matd\rot\vcan\}_{\rot}$
		\end{tabular}		
	};
	
	\node (unq) [startstopflex, xshift=6cm, yshift=-2cm] {
		\begin{tabular}{c}
		Unique Expression: \\ 
		$\matd\vegn$~~[Eq.~\ref{eq:xi-zeta-I-final}]
		\end{tabular}		
	};
	
	\node (b1) [sum, xshift=1cm, yshift=-2cm] {};
	
	\draw [arrow] (intr) -| 
			node[xshift=2cm, anchor=center, above, text width=2.1cm] {expressed in}  
			node[xshift=2cm, anchor=center, below, text width=2.2cm] {canonical basis} 
	(many); % slanted left arrow
		
	\node (bsEgn) [block, xshift=-1.8cm, yshift=-3.4cm] {
		\begin{tabular}{c}
		Optimal POVM: \\ 
		$\vegn = \rot\vcan$
		\end{tabular}		
	}; 
	
	\node (b2) [xshift=-5.4cm, yshift=-3.4cm] {};
	
	\node (gam) [block, xshift=-6cm, yshift=-5cm] {$\Gamma_{xy}$ over $\vcan$};
	
	\node (gamDscr) [block, xshift=3.4cm, yshift=-5cm] {		
		\begin{tabular}{c}
		Unique Expression~[Eq.~\ref{eq:Gxy:inEigenbasis}] over $\vegn$: \\
		\small$\begin{matrix}
		\half(\Dp_{uv}^2-\Dm_{uv}^2)\left[\left(1-D_{xy}\right)\!\left(\EE_{00}-\EE_{11}\right) + D_{xy}\!\left(\EE_{22}-\EE_{33}\right)\right]
		\end{matrix}$ 
		\end{tabular}	
	};	
	
	\node (b3) [sum, xshift=-2.2cm, yshift=-5cm] {};
	
	\draw [arrow] (many) -- node[] {} (b1);	
		
	\draw [arrow] (b1) -- 
		node[anchor=center, above, text width=2.1cm] {expressed in}  
		node[anchor=center, below, text width=2.6cm] {eigenbasis of $\Gamma_{xy}$}
	(unq);
	
	\draw [arrowDashed] (bsEgn) -- node[] {} (b3); % arrow
	\draw [arrowDashed] (gam) -- node[] {} (b3); % arrow
	\draw [arrowDashed] (b3) -- node[] {} (gamDscr); % arrow
	
	\draw [arrowDashed] (many) -- node[] {} (gam); %  arrow	
	
	\draw [arrowDashed] (gam.north) |- ++(right:0.5) |- node[anchor=center, below, text width=1.5cm, xshift=1.2cm] {eigenbasis} (bsEgn.west);
	
	\draw [arrowDashed] (bsEgn) -| node[]{}(b1); % right arrow
	
	\end{tikzpicture}
}

\newcommand*{\tikzUnq}{     
	\begin{tikzpicture}[node distance=2cm] 
		
	\node (many) [block, xshift=3cm, yshift=0cm]  {
		\begin{tabular}{c}
		$\{\matd\rot\vcan\}_{\rot}$
		\end{tabular}
	};
	
	\node (unq) [startstop, left of=many, xshift=-3cm, yshift=1cm] {
		\begin{tabular}{c}
		Optimal Interaction \smallskip\\ 
		\footnotesize{Unique Expression} \\
		$\matd\vegn$~~[Eq.~\ref{eq:xi-zeta-I-final}]
		\end{tabular}
	};
	
	\node (optPOVM) [startstop, left of=many, xshift=-3cm, yshift=-1cm] {
		\begin{tabular}{c}
		Optimal POVM \smallskip\\ 
		$\vegn = \rot\vcan$~~[Eq.~\ref{eq:optimal_POVM-vec}]
		\end{tabular}
	};
	\node (b1) [sum, left of=many, xshift=-0.8cm, yshift=0cm] {};
	
	\draw [arrow] (unq) -- node[]{}  (b1); 
	\draw [arrow] (optPOVM) -- node[]{}  (b1); 
	\draw [arrow] (b1) -- node[]{}  (many); 
	\node (POVM1) [rectangle, draw=black, right of=many, xshift=1cm, yshift=1.5cm] {POVM 1 [Eq.~\ref{eq:optEl}]};
	\node (POVM2) [rectangle, draw=black, right of=many, xshift=1cm, yshift=0.5cm] {POVM 2 [Eq.~\ref{eq:optimal_POVM-vec:eqErr}]};
	\node (POVM3) [rectangle, draw=black, right of=many, xshift=1cm, yshift=-0.5cm] {POVM 3 [Eq.~\ref{eq:optimal_POVM-vec_exmpl}]};
	\node (POVM4) [rectangle, draw=black, right of=many, xshift=1cm, yshift=-1.5cm] {POVM 4 [Eq.~\ref{eq:optimal_POVM-vec-example1}]};
	\draw [arrow] (many) -- (POVM1.west); % right arrow
	\draw [arrow] (many) -- (POVM2.west); % right arrow
	\draw [arrow] (many) -- (POVM3.west); % right arrow
	\draw [arrow] (many) -- (POVM4.west); % right arrow
	
	\node (int1) [rectangle, draw=black, right of=POVM1, xshift=2cm] {Interaction 1 [Eq.~\ref{eq:opt_xi-zeta}]};
	\node (int2) [rectangle, draw=black, right of=POVM2, xshift=2cm] {Interaction 2 [Eq.~\ref{eq:choice:xi_zeta:opt}]};
	\node (int3) [rectangle, draw=black, right of=POVM3, xshift=2cm] {Interaction 3 [Eq.~\ref{eq:xi-zeta-opt-gen}]};
	\node (int4) [rectangle, draw=black, right of=POVM4, xshift=2cm] {Interaction 4 [Eq.~\ref{eq:xi-zeta-opt-gen-example1}]};
	\draw [arrow] (POVM1) -- node[] {} (int1); % right arrow
	\draw [arrow] (POVM2) -- node[] {} (int2); % right arrow
	\draw [arrow] (POVM3) -- node[] {} (int3); % right arrow
	\draw [arrow] (POVM4) -- node[] {} (int4); % right arrow
	\end{tikzpicture}
}

\begin{document}
%\title{Revisiting optimal eavesdropping in quantum cryptography: \\Choice of interaction is unique up to rotation of the underlying basis}
\title{Revisiting optimal eavesdropping in quantum cryptography: \\Optimal interaction is unique up to rotation of the underlying basis}
\author{Atanu Acharyya}
\email{pub.academy.15@gmail.com}
\affiliation{Applied Statistics Unit, Indian Statistical Institute, Kolkata 700 108, India}
\author{Goutam Paul}
\email{goutam.paul@isical.ac.in}
\affiliation{Cryptology and Security Research Unit, Indian Statistical Institute, Kolkata 700 108, India}

\begin{abstract}
A general framework of optimal eavesdropping on BB84 protocol was provided by Fuchs \etal. [\hyperref{http://journals.aps.org/pra/abstract/10.1103/PhysRevA.56.1163}{category}{name}{Phys. Rev. A, 1997}]. An upper bound on mutual information was 
derived, which could be achieved by a specific type of interaction and the corresponding measurement. However, uniqueness of optimal interaction was posed as an unsolved problem there and it has remained open for almost two decades now. In this paper, we solve this open problem and establish the uniqueness of optimal interaction up to rotation. The specific choice of optimal interaction by Fuchs \etal. is shown to be a special case of the form derived in our work. 
\end{abstract}
\maketitle

\section{Introduction}
\label{intro}
Symmetric key cryptography requires a secret key to be shared or distributed between the sender (say, Alice) and the receiver (say, Bob). The security of classical key distribution is based on hardness assumptions for solving certain computational problems. This gives security against computationally bounded adversary in the classical domain, but fails to guarantee security against quantum attacks. Quantum key distribution (QKD) is based on the principles of quantum mechanics. 
To encode classical bits, QKD uses quantum states which the attacker (say, Eve) cannot measure without creating disturbance detectable by Bob. QKD protocol does not require any computation complexity assumption and is provably secure against both classical as well as quantum adversaries. 

The first and possibly the most celebrated QKD protocol is BB84~\cite{bb84}. The protocol relies on the use of orthogonal states from one of the two conjugate bases, say, $\xy$ and $\uv$, to encode a bit-string in qubits (e.g., polarized photons). Alice randomly selects one of the two bases and encodes 0 and 1 respectively by a qubit prepared in one of the two states in each base. Say, Alice encodes 0 to $\ket{x}$ or $\ket{u}$, and 1 to $\ket{y}$ or $\ket{v}$, depending on the chosen basis. When Bob receives a state from Alice, he randomly selects a basis $\xy$ or $\uv$ and makes a measurement. Once the measurement is done for all the received qubits, Alice and Bob publicly announce the sequence of bases used by them and discard the bits where the bases do not match. The resulting bit string, followed by error correction and privacy amplification, becomes the common secret key. However, presence of an eavesdropper may disturb the state of a qubit sent by Alice for which Bob may get a wrong result even if the corresponding bases of measurement between Alice and Bob match. To overcome this problem, Alice and Bob sacrifice some of the bits by comparing their values publicly.

Fuchs \etal.~\cite{fuchs97} provided a general framework of optimal eavesdropping on BB84 protocol. They derived an upper bound on mutual information, described a specific type of interaction and the corresponding measurement that achieves the bound. They finally explained an optimal strategy for Eve in interpreting her measurement. However, the optimal interaction described there was a specific choice and the uniqueness of the optimal interaction was left as an open problem. They commented:
%\begin{quote}
	``{\em It is easy to check that the solution here is correct, but the extent to which it is unique aside from trivial changes of basis and of phase͒ remains unknown}."
%\end{quote}

Interestingly, this problem has been open for last two decades. In this paper, we solve this open problem and establish the uniqueness up to rotation of the underlying basis.
%Here we do not claim any improvement of optimal information gain over~\cite{fuchs97}. 
We characterize the classes of interaction that can achieve the already-existing optimal bound given by~\cite{fuchs97}.
We have shown that the choice of optimal interaction in~\cite{fuchs97} is a special case of the generalized form provided by us. We also explicitly show the corresponding optimal measurement by Eve. 

Note that Fuchs \etal.~\cite{fuchs97} made an intelligent guess to arrive at the expression for optimal iteraction. On the other hand, in this paper, we explicitly derive the general form of the expression of any possible optimal interaction. See Sec.~\ref{connectfuchs} for a more elaborate discussion on this issue.

The content of this paper is organized as follows. Section~\ref{prelim} explains basic terminologies used for optimal eavesdropping introduced in~\cite{fuchs97}. Section~\ref{revisit} contains summary of certain results from~\cite{fuchs97} which are relevant to our work. Our results are explained in Sec.~\ref{ourResults}. The remaining portion discusses the connection of our results with~\cite{fuchs97} followed by a conclusion.

\section{Preliminaries}
\label{prelim}
Alice and Bob want to share a secret key using BB84 protocol. Alice randomly chooses a basis from   $\fB_{xy} = \{\ket{x}, \ket{y}\}$ and $\fB_{uv} = \{\ket{u}, \ket{v}\}$,
where
\begin{equation} \label{eq:xuv}
\ket{x} = \xuvExpr{u}{v}, \indent \ket{y} = \yuvExpr{u}{v},
\end{equation}
i.e., the bases are conjugate to each other.  \footnote{Note that a more common notation uses $\ket{0}$, $\ket{1}$, $\ket{+}$ and $\ket{-}$ instead of $\ket{x}$, $\ket{y}$, $\ket{u}$ and $\ket{v}$ respectively. However, we follow the same notations as in Fuchs \etal.~\cite{fuchs97} so that the connection to their work is easily visible.}

Alice encodes her key-bits, each as a polarized photon, and sends it to Bob. Suppose, an eavesdropper Eve interferes the communication while she lets a probe interact unitarily with the qubit sent by Alice. 

Suppose Alice has chosen a signal, say, $\ket{x}$ (corresponding density operator being $\rho_x^A = \ketbra{x}{x}$ ), in the basis  $\fB_{xy}$. Eve lets a probe, initially in state $\ket{\psi_0}$ (corresponding density operator $\rho_{0}^E = \ketbra{\psi_0}{\psi_0}$), interact unitarily (realized by a unitary operator $\U$) with the qubit sent by Alice. The post-interaction joint state $\ket{X}$ between Alice and Eve, which is an entangled state of the probe of Eve and the photon sent by Alice, is realized by
\begin{equation*} 
\ket{x}\tensorprod\ket{\psi_0}\xrightarrow{~\U~}~\ket{X}. 
\end{equation*}
Bob receives a simple mixture of the two basis vectors (here $\fB_{xy}$) chosen by Alice, i.e., Bob's density matrix is always diagonal in the basis chosen by Alice. Thus, Schmidt decomposition of the post-interaction joint state $\ket{X}$ must be of the form
\begin{equation*} 
\ket{X} = \entangledExprGen{\a}{x}{x}{y}{x}, 
\end{equation*} such that
\begin{equation} 
\label{rel:perp_x}
\ket{\xi_x} \perp \ket{\zeta_x},
\end{equation}
where $\ket{\xi_x}, \ket{\zeta_x}$ are component of Eve's part of the joint state after the interaction.

Similarly, when Alice sends $\ket{y}$, the post-interaction state $\ket{Y}$ must be of the form
\begin{equation*} 
\ket{Y} = \entangledExprGen{\b}{y}{y}{x}{y}, 
\end{equation*} such that
\begin{equation} 
\label{rel:perp_y}
\ket{\xi_y} \perp \ket{\zeta_y}.
\end{equation}
The density operator for the post-interaction state $\ket{X}$ is given by 
\begin{eqnarray} 
\label{eq:op-joint-state}
\rho_x^{AE} = \ketbra{X}{X} =  \U \left(\rho_x^A\tensorprod \rho_{0}^E\right) \U^\dagger.
\end{eqnarray} 
Eve's description of the system will be
\footnote{Henceforth, we use the notation $:=$ to denote ``defined as".}
\begin{eqnarray} \label{eq:rho_x} 
\rho_x &:=& \rho_x^{E} ~=~ \tr_A\left(\rho_x^{AE}\right) ~=~ \tr_A\left(\ketbra{X}{X}\right),
\end{eqnarray}
where $\tr_A$ represents partial trace over Alice's qubit.

Since the interaction is unitary, it follows from Eqs.~\mref{eq:xuv,eq:op-joint-state} that
\begin{eqnarray}
\label{eq:XUV}
\ket{X} = \xuvExpr{U}{V}, \indent \ket{Y} = \yuvExpr{U}{V}.
\end{eqnarray}

Before performing any measurement, Eve waits until Alice declares her choice of basis publicly. Eve's measurement is considered to be a Positive Operator-Valued Measure (POVM) $\{E_\l\}$ or $\{F_\l\}$ depending on whether Alice's choice is $x\mhyphen y$ or $u\mhyphen v$ basis. Note that the operators $\{E_\l\}$ satisfy two properties~\cite{fuchs96, qNC02}:
they are all non-negative definite, i.e., 
\begin{eqnarray*}
	\braopket{\gamma}{E_\l}{\gamma} &\ge& 0,~~\forall ~\ket{\gamma}, 
\end{eqnarray*}
and satisfy the completeness relation
\begin{eqnarray*}
	\sum\limits_\l E_\l &=& \mathbbm{1}.
\end{eqnarray*}

Suppose, Alice sends a signal in $\xy$ (or, $\uv$) basis with 
the prior probabilities $p_x, p_y$ (or, $p_u, p_v$) respectively. Once Alice reveals her basis to be $x\mhyphen y$, Eve uses a POVM $\{E_\l\}$ to perform a measurement on her probe. Considering $\A,\B,\E$ as random variables corresponding to the signal sent by Alice, signal received by Bob, and, measurement outcome of Eve, the conditional probability of the various outcomes $\l$ of that measurement is given by
\begin{eqnarray}
P_{\l x} &:=& \Pr[\E=\l | \A=x] = \tr\left(\rho_xE_\l\right),  \tab \label{eq:PlxTr} \\
P_{\l y} &:=& \Pr[\E=\l | \A=y] = \tr\left(\rho_yE_\l\right).  \tab \label{eq:PlyTr}  
\end{eqnarray} 

The probability that Eve gets outcome $\l$, when Alice uses $x\mhyphen y$ basis is thus
\begin{eqnarray*}
	q_{xy}(\l) &:=& \Pr[\E=\l] = P_{\l x}p_x + P_{\l y}p_y. 
\end{eqnarray*} 
Looking at outcome $\l$, Eve assigns a guess for the signal sent by Alice following some strategy. The posterior probability $Q_{x\l}$ (or $Q_{y\l}$) of the event that Alice had sent $x$ (or $y$) given that Eve has observed $\l$ is given by Bayes' theorem.
\begin{eqnarray*}
	Q_{x\l} &:=& \Pr[\A=x | \E=\l] = \frac{P_{\l x}p_x}{q_{xy}(\l)}, \\
	Q_{y\l} &:=& \Pr[\A=y | \E=\l] = \frac{P_{\l y}p_y}{q_{xy}(\l)}.
\end{eqnarray*}
A simple way that Eve can utilize these likelihoods is to perform a guess realized by the following function.
\vspace{0em}\[ \argmax\:\{Q_{x\l}, Q_{y\l}\} =   %--- \: is a short space
\begin{cases}	
x \text{, ~~ if } Q_{x\l} > Q_{y\l}, \\
y \text{, ~~ if } Q_{y\l} > Q_{x\l}.
\end{cases}
\] 
A convenient measure of Eve's \textbf{information gain} for an outcome $\l$, as proposed in \cite{fuchs97}, is
\begin{eqnarray*}
	G_{xy}(\l) &:=& \left| Q_{x\l} - Q_{y\l} \right|. 
\end{eqnarray*}	
On average, Eve's information gain over all outcomes is
\begin{eqnarray*}
	G_{xy} &:=& \sum\limits_\l q_{xy}(\l)G_{xy}(\l) = \sum\limits_\l \left|P_{\l x}p_x-P_{\l y}p_y\right|.
\end{eqnarray*}
In particular, for equiprobable signals, 
\begin{eqnarray*} 
	G_{xy} &=& \half \sum\limits_\l \left|P_{\l x}-P_{\l y}\right|. 
\end{eqnarray*}	
A more sophisticated data processing by Eve is \textbf{mutual information} \cite{cover91}. For equal prior, this is given by
\begin{eqnarray*}
	I_{xy} &:=& \ln2+ \sum\limits_{\l} q_{xy}(\l) \left(Q_{x\l} \ln Q_{x\l} + Q_{y\l} \ln Q_{y\l} \right).
\end{eqnarray*}

Eve's attempt to measure the probe creates \textbf{disturbance} to the signal sent by Alice which is detectable by Bob. For signal sent in $\xy$ basis, the disturbance incorporated by Eve could be described by 
\begin{eqnarray*}
	D_{xy} &:=& \sum\limits_\l q_{xy}(\l) d_{xy}(\l),
\end{eqnarray*}
where, $d_{xy}(\l)$ is the avg error for Bob to read the signal sent by Alice while Eve detects $\l$. For equal prior, 
\begin{eqnarray*}
	d_{xy}(\l) &:=& \half \left(d_{\l x} + d_{\l y}\right),
\end{eqnarray*}
where, $d_{\l x}$ is the error for Bob when Alice sends $x$ while Eve detects $\l$ (i.e., Bob reads $y$), i.e., 
\begin{eqnarray*}
	d_{\l x} &:=& \Pr[\B=y|(\A=x,\E=\l)],
\end{eqnarray*}
and $d_{\l y}$ is the error for Bob when Alice sends $y$ while Eve detects $\l$ (i.e., Bob reads $x$), i.e., 
\begin{eqnarray*}
	d_{\l y} &:=& \Pr[\B=x|(\A=y,\E=\l)].
\end{eqnarray*}
Clearly, $D_{xy}$ is the observable error rate of Bob to read the signal sent by Alice prepared in $\xy$ basis.

Similarly, one can define $G_{uv},I_{uv},D_{uv}$ while considering Alice's signal was prepared in $\uv$ basis. We drop the subscripts $xy$ and $uv$, i.e., use $G,I$, when both the bases to be considered in discussion. 

\section{Summary of Optimal Eavesdropping by Fuchs \etal.~\cite{fuchs97}}
\label{revisit}
Optimal eavesdropping means that an eavesdropper performs the interaction and the measurement in such a way that she can extract maximum information about the signal sent by Alice, ensuring that the disturbance at Bob's end remains bounded by a suitable threshold. In the QKD literature, it is interpreted as maximizing the information gain by Eve or mutual information between Alice and Eve. For BB84 protocol, considering the interaction to be unitary and restricted to equal prior ($p_x=\half=p_y$), Fuchs \etal.~\cite{fuchs97} provided an upper bound on information gain and mutual information over all possible interaction-POVM pairs. A criterion to achieve the bounds was provided there. To show that these bounds are attainable, an interaction-POVM pair for unequal error rates and another for equal error rates were provided therein. These results are discussed briefly in this section. Since these results hold for equal prior, the subsequent sections follow the same assumption unless explicitly mentioned.

%\footnote{For any quantity $q$, we denote its optimal (maximum) value by $q^{\star}$.}
\subsection{Upper bounds on information gain ($G$) and mutual information ($I$)}
\label{bound}
For equal prior, Fuchs \etal. \cite{fuchs97} provided an upper bound on the information gain ($G$). This bound was used to provide an upper bound on the mutual information ($I$). A necessary and sufficient condition to achieve the bounds was given there. We recollect these results here.

\begin{proposition}{\em({\it An upper bound on information gain ($G$)})~\cite[Eqs.~(23,24)]{fuchs97}}
% \\ For a given POVM $\{E_\l\}$,	
		\begin{eqnarray} 
		G_{xy} &\le& 2 \concave{D_{uv}} \label{eq:boundG_xy}, 
		\end{eqnarray}	
		\begin{eqnarray} 
		G_{uv} &\le& 2 \concave{D_{xy}} \label{eq:boundG_uv}.
		\end{eqnarray}
		Moreover, for measurement outcome $\l$ of Eve, the bound on information gain~\cite[Eq.~(20)]{fuchs97} can be expressed by the following inequality	
		\begin{eqnarray} 
		G_{xy}(\l) &\le& 2 \concavesq{d_{uv}(\l)}. \label{eq:boundGl_xy}
		\end{eqnarray}
\end{proposition}
It it interesting to note that while Eve's information gain refers to signals sent in the $x\mhyphen y$ basis, Bob's error rate refers to signals sent in the $u\mhyphen v$ basis and vice versa.

\begin{proposition}{\em({\it An Upper Bound on Mutual Information ($I$)})~\cite[Eqs.~(31,32)]{fuchs97}}
	%\\ For a given POVM $\{E_\l\}$,	 	
	\begin{eqnarray} \label{eq:boundMI_xy}
	I_{xy} &\le& \half~\phiconcave{D_{uv}},
	\end{eqnarray}	
	\begin{eqnarray} \label{eq:boundMI_uv}
	I_{uv} &\le& \half~\phiconcave{D_{xy}},
	\end{eqnarray}
	where $\phi(z) = \philog{z}$.
\end{proposition}
Subscripts emphasize that the mutual information and error rate refer to signals sent in two different bases.

\begin{proposition}
\label{condns:optG} 
{\em({\it Necessary and Sufficient Conditions to Achieve $G^{\star}$})\footnote{$q^{\star}$ denotes optimal (maximum) value for any quantity $q$.}~\cite[Eqs.~(38,39)]{fuchs97}}\\ 
The necessary and sufficient conditions for equality in Eq.~\eqref{eq:boundG_xy} are 
		\begin{eqnarray} \label{eq:condn_optG_1}
		\ket{V_{\l u}} &=& \varepsilon_\l~\factor{uv}~ ~\ket{U_{\l u}} 
		\end{eqnarray}
		and 
		\begin{eqnarray} \label{eq:condn_optG_2}
		\ket{U_{\l v}} &=& \varepsilon_\l~\factor{uv}~ ~\ket{V_{\l v}}, 
		\end{eqnarray}
		where 
		\begin{eqnarray} \label{eq:sign:el}
		\varepsilon_\l &=& \pm1 = \sgn\left(Q_{x\l} - Q_{y\l}\right) 
		\end{eqnarray}
		and 
		\begin{eqnarray} \label{eq:Ulu}
		\ket{U_{\l u}} = \Xlambdax{U}{u}, \tab \ket{V_{\l u}} = \Xlambdax{V}{u}, \nonumber\\ 
		\ket{U_{\l v}} = \Xlambdax{U}{v}, \tab \ket{V_{\l v}} = \Xlambdax{V}{v}, \nonumber\\
		%\multirow{1}{*}{3} 
		B_u = \ketbra{u}{u}, ~~ B_v = \ketbra{v}{v}, ~~~\text{ with }~~ B_u + B_v = \mathbbm{1}. ~~\nonumber\\
		\end{eqnarray} 
		Similar conditions hold for a signal prepared in $\uv$ basis to attain the equality in Eq.~\mref{eq:boundG_uv}. 
	\end{proposition}
It is intriguing to note that the set of conditions that optimizes $G$ also optimizes $I$. Therefore, the necessary and sufficient conditions for equality in Eqs.~\mref{eq:boundMI_xy,eq:boundMI_uv} are also the same as those in Proposition~\ref{condns:optG}. That is to say, for a signal sent in $\xy$ basis, an interaction-POVM pair that attains the bound in Eq.~\mref{eq:boundG_xy} does the same in Eq.~\mref{eq:boundMI_xy} and vice versa. For the other basis, similar statement holds for Eqs.~\mref{eq:boundG_uv} and~\mref{eq:boundMI_uv}.
%	Eqs.~\mref{eq:boundG_xy,eq:boundG_uv} should also attain the bounds in Eqs.~\mref{eq:boundMI_xy,eq:boundMI_uv} and vice versa.} 
%This issue will be encountered in Sec.~\mref{optPOVM, opttimality:check} again.

\subsection{Description of the postinteraction states $\ket{X}, \ket{Y}$ \label{postIntrSt}}
Eve's objective is to maximize $G$ or $I$, irrespective of what basis was used by Alice for encoding. Both the bounds~\mref{eq:boundMI_xy,eq:boundMI_uv} [and therefore the bounds~\mref{eq:boundG_xy,eq:boundG_uv}] could be achieved simultaneously while fixing $D_{xy}, D_{uv}$ independently~\cite{fuchs97}. One of the conditions that must hold to achieve the bounds in $\xy$ basis is the following~\cite[Eq.~(33)]{fuchs97}: 
\begin{equation*} d_{\l u} = d_{\l v} = d_{uv}(\l) = D_{uv}, ~~\forall \l .\end{equation*}
A similar condition holds good for signals sent in $\uv$ basis.

Thus, for a signal sent in $\xy$ basis, the Schmidt decomposition of the postinteraction states are
\begin{eqnarray}  \label{eq:entangledExprXY}
\ket{X} &=& \entangledExpr{xy}{x}{x}{y}{x}, \nonumber \\
\ket{Y} &=& \entangledExpr{xy}{y}{y}{x}{y}. 
\end{eqnarray} 
Assuming that all inner products $\braket{\xi_i}{\zeta_j}$ are real, the restrictions~\mref{rel:perp_x,rel:perp_y} on $\ket{\xi_i}, \ket{\zeta_j}$ becomes more restricted as
\begin{eqnarray} \label{rel:perp_opt}
\{\ket{\xi_x},\ket{\xi_y}\} \perp \{\ket{\zeta_x}, \ket{\zeta_y}\}. 
\end{eqnarray}
Similarly, for a signal sent in $\uv$ basis, the post-interaction states are
\begin{eqnarray} \label{eq:entangledExprUV}
\ket{U} &=& \entangledExpr{uv}{u}{u}{v}{u}, \nonumber \\
\ket{V} &=& \entangledExpr{uv}{v}{v}{u}{v}. 
\end{eqnarray}
Since the bases $\fB_{xy}$ and $\fB_{uv}$ are conjugate to each other, we expect to get a relationship between $\ket{\xi_i},\ket{\zeta_j}$ in $\uv$ basis and those in $\xy$ basis which is described below.
\newcommand*{\Scale}[2][4]{\scalebox{#1}{$#2$}} %--- reduce font size in math eqn
\begin{eqnarray} \label{eq:xiuxy}
\Scale[0.98]{
2\sqrt{1\!-\!D_{uv}}\ket{\xi_u} \!=\! \sqrt{1\!-\!D_{xy}}(\ket{\xi_x}\!+\!\ket{\xi_y}) \!+\!%
\sqrt{D_{xy}}(\ket{\zeta_x}\!+\!\ket{\zeta_y}),} \nonumber\bigskip\\
\Scale[0.97]{
2\sqrt{D_{uv}}\ket{\zeta_u} \!=\! \sqrt{1\!-\!D_{xy}}(\ket{\xi_x}\!-\!\ket{\xi_y}) \!+\!%
\sqrt{D_{xy}}(\ket{\zeta_y}\!-\!\ket{\zeta_x}).} \nonumber\\
\end{eqnarray}
Similarly, 
\begin{eqnarray} \label{eq:xivxy}
\Scale[0.97]{
	2\sqrt{1\!-\!D_{uv}}\ket{\xi_v} \!=\! \sqrt{1\!-\!D_{xy}}(\ket{\xi_x}\!+\!\ket{\xi_y}) \!-\!%
	\sqrt{D_{xy}}(\ket{\zeta_x}\!+\!\ket{\zeta_y}),} \nonumber\bigskip\\
\Scale[0.97]{
	2\sqrt{D_{uv}}\ket{\zeta_v} \!=\! \sqrt{1\!-\!D_{xy}}(\ket{\xi_x}\!-\!\ket{\xi_y}) \!-\!%
	\sqrt{D_{xy}}(\ket{\zeta_y}\!-\!\ket{\zeta_x}).} \nonumber\\
\end{eqnarray}

From the orthogonality relation~\eqref{rel:perp_opt}, one can say that 
Eve's probe lives in a Hilbert space of dimension at most four, and thus is taken to be made of 2 qubits (4 states). It is therefore convenient to introduce same bases ($x\mhyphen y$ and $u\mhyphen v$, used by Alice) for each of Eve's qubits.

\subsection{Optimal interaction to maximize $G,I$: \\A specific choice \label{optPOVM-choice}}
Any interaction, as described above, that leads to optimality (i.e., attains $G^{\star}$ or $I^{\star}$) could be chosen. In~\cite[Sec.~III: Eqs.~(50,51)]{fuchs97}, one such specific choice was made for unequal error rates, which was shown to be a correct choice (correct in the sense that the choice leads to optimality). Similarly, for equal error rates, another specific choice was made in \cite[Sec.~IV, Eq.~(69)]{fuchs97}. However, uniqueness of the choice was left as an open problem in~\cite[Sec.~III, first paragraph]{fuchs97}.

\subsubsection{\textbf{For unequal error rates, i.e.,} $D_{xy} \ne D_{uv}$}
\label{optPOVM-choice-uneqD}
Equations~(50,~51) of \cite[Sec.~III]{fuchs97} are restated here.
Consider a canonical basis for Eve's probe as $\{\ket{\E_0}, \ket{\E_1}, \ket{\E_2}, \ket{\E_3}\}$. Without loss of generality, % (w.l.o.g.)
\begin{equation}
\label{eq:basis_cano_Eve}
\ket{\E_0}=\txx, \ket{\E_1}=\tyx, \ket{\E_2}=\txy, \ket{\E_3}=\tyy.
\end{equation}
To describe $\ket{\xi_i}, \ket{\zeta_j}$, the work~\cite{fuchs97} considered an orthonormal set, namely, the Bell Basis with respect to (w.r.t.) $\xy$, as follows.
\begin{eqnarray} 
\label{eq:phi_xy}
\Phixy{\pm} &:=& \PhixyExpr{\pm} = \invsqrttwo\left(\ket{\E_0}\pm\ket{\E_3}\right), \nonumber \\
\Psixy{\pm} &:=& \PsixyExpr{\pm} = \invsqrttwo\left(\ket{\E_2}\pm\ket{\E_1}\right). \nonumber \\
\end{eqnarray}
\noindent In terms of the Bell basis vectors for Eve's probe, the interaction was chosen such that
\begin{eqnarray} 
\label{eq:optimal_xi-zeta}
\ket{\xi_x} &=& \sqrt{1-D_{uv}}~\Phixy{+} + \sqrt{D_{uv}}~\Phixy{-}, \nonumber\\
\ket{\xi_y} &=& \sqrt{1-D_{uv}}~\Phixy{+} - \sqrt{D_{uv}}~\Phixy{-}, \nonumber\\
\ket{\zeta_x} &=& \sqrt{1-D_{uv}}~\Psixy{+} - \sqrt{D_{uv}}~\Psixy{-}, \nonumber\\
\ket{\zeta_y} &=& \sqrt{1-D_{uv}}~\Psixy{+} + \sqrt{D_{uv}}~\Psixy{-}. 
\end{eqnarray}
The corresponding optimal POVM, as shown in \cite[Eqs.~(55,56)]{fuchs97},
is described below.
\begin{equation*}
\El = \ketbra{\El}{\El},
\end{equation*}
where 
%\begin{equation}\label{eq:optEl}
%\ket{E_0}=\ket{\E_0}, ~~\ket{E_1}=\ket{\E_1}, ~~\ket{E_2}=\ket{\E_2}, ~~\ket{E_3}=\ket{\E_3}.
%\end{equation}
\begin{equation}\label{eq:optEl}
\ket{E_0}\!=\!\ket{\E_0}, ~\ket{E_1}\!=\ket{\E_1}, ~\ket{E_2}\!=\!\ket{\E_2}, ~\ket{E_3}\!=\!\ket{\E_3}.~~
\end{equation} 

Introducing new notations $\Dp_{uv}, \Dm_{uv}$, we can write a closed form
of $\ket{\xi_i}, \ket{\zeta_j}$ as below.
	\begin{eqnarray} \label{eq:opt_xi-zeta}
	\ket{\xi_x} &=& \Dp_{uv}~\ket{\E_0} + \Dm_{uv}~\ket{\E_3}, \nonumber\\
	\ket{\xi_y} &=& \Dm_{uv}~\ket{\E_0} + \Dp_{uv}~\ket{\E_3}, \nonumber\\
	\ket{\zeta_x} &=& \Dm_{uv}~\ket{\E_2} + \Dp_{uv}~\ket{\E_1}, \nonumber\\
	\ket{\zeta_y} &=& \Dp_{uv}~\ket{\E_2} + \Dm_{uv}~\ket{\E_1},
	\end{eqnarray} 
	where 
	\begin{eqnarray}  \label{rel:DpDm}
	\Dp_{uv} & := & \frac{\sqrt{1-D_{uv}}+\sqrt{D_{uv}}}{\sqrt{2}},\nonumber\\
\Dm_{uv} & := & \frac{\sqrt{1-D_{uv}}-\sqrt{D_{uv}}}{\sqrt{2}}. 
	\end{eqnarray}
The following relations appear to be useful.
		\begin{eqnarray}\label{rel:calD_uv} 
		\Dp_{uv}\cdot\Dm_{uv} & = & \half\left(1-2D_{uv}\right),\nonumber\\
\Dp_{uv}^2+\Dm_{uv}^2 & = & 1, \nonumber\\
		\Dp_{uv}^2-\Dm_{uv}^2 & = & 2\concave{D_{uv}}.
		\end{eqnarray}
The above analysis works for a signal sent in $\xy$ basis. Similar analysis holds for $\uv$ basis as well.

\subsubsection{\textbf{For equal error eates, i.e., $D_{xy}=D_{uv}=D$} \label{optPOVM-choice-eqD}}
%\begin{paragraph}{For Equal Error Rates, i.e., $D_{xy}=D_{uv}=D$:}
For equal error rates, \cite[Sec.~IV, Eq.~(69)]{fuchs97} comes up with another choice of $\ket{\xi_i}, \ket{\zeta_j}$. We describe it as below.
\begin{eqnarray} \label{eq:choice:xi_zeta}
\ket{\xi_x} &=& \txx,  \nonumber\\
\ket{\xi_y} &=& \xyscaled{\alpha} \ket{x} \nonumber\\
\ket{\zeta_x} &=& \txy,  \nonumber\\
\ket{\zeta_y} &=& \xyscaled{\beta}\ket{y}.
\end{eqnarray}
Optimality of $G$ (or $I$) is reached when 
\begin{eqnarray*}
\alpha = \beta &\text{~~and~~}& \sin\alpha = 2\concave{D} = \Dp^2-\Dm^2.
\end{eqnarray*}
Here, the notations $\Dp,\Dm$ are analogous to $\Dp_{uv},\Dm_{uv}$ in Eq.~\eqref{rel:DpDm} but for equal error rates, i.e., to consider $D$ than $D_{uv}$ for the right hand side of Eq.~\eqref{rel:DpDm}.
%The notations $\Dp,\Dm$ have already been defined in Eq.~\eqref{rel:DpDm}.

Thus, the optimal interaction can be written as 
\begin{eqnarray} \label{eq:choice:xi_zeta:opt}
\ket{\xi_x} &=& \ket{\E_0}, \nonumber\\
\ket{\xi_y} &=& 2\Dp\Dm~\ket{\E_0}+\left(\Dp^2-\Dm^2\right)\ket{\E_1}, \nonumber\\
\ket{\zeta_x} &=& \ket{\E_2}, \nonumber\\
\ket{\zeta_y} &=& 2\Dp\Dm~\ket{\E_2}+\left(\Dp^2-\Dm^2\right)\ket{\E_3}.
\end{eqnarray}
However, the corresponding optimal POVM was not shown explicitly in \cite{fuchs97}, which we establish in Sec.~\ref{optPOVM_eqD}.

Although, both interactions~\mref{eq:opt_xi-zeta,eq:choice:xi_zeta:opt} lead to optimality, the way they were proposed in~\cite{fuchs97} seems to be an intelligent guesswork. This leaves open a few interesting questions:
\begin{enumerate}[noitemsep]
	\item Instead of guessing an interaction and verifying its optimality, can one derive it from first principle?
	\item Are there other possible optimal interactions than the two specific ones?
	\item If so, is it possible to characterize them? 
\end{enumerate}
We address these questions in the following section.
	
%	, thereby leaving a scope to derive the family of interaction-POVM pairs that leads to optimality. In the following section, we perform this derivation while we encounter an unique representation for the entire family.}

\section{Our Results}
\label{ourResults}
In this section, we derive an expression for an interaction by Eve that leads to optimal information gain. Eventually, we show that the expression is unique in a fixed basis. Associated optimal POVMs are then identified.

Given an interaction, how to identify an optimal POVM is discussed in the following subsection.
%\clrA{Given an interaction, identification of an optimal POVM is well known~\cite{fuchs96}. This is the subject of the following subsection.} 
\subsection{Optimal measurement (POVM) to maximize information gain ($G$) for a given interaction}
\label{optPOVM}
Let's consider the problem below: given an interaction,%given a fixed distribution $(p_{x}, p_{y})$ set for the communication, and for a given interaction $(\rho_x, \rho_y)$ by Eve,
\begin{eqnarray*}
	\text{maximize } G_{xy} &=& \sum\limits_\l \left|P_{\l x}p_x-P_{\l y}p_y\right|
\end{eqnarray*}
over all POVMs $\{E_\l\}$.	

In~\cite{fuchs96}, an optimal measurement for this maximization was derived.
There, the maximization was done on {\em Kolmogorov Variational Distance}~\cite[Eq.~(130)]{fuchs96}. The calculation was performed 
in~\cite[Appendix (Sec.~7)]{fuchs96}, which shows that the optimal measurement corresponds to a Hermitian operator given by~\cite[Eq.~(21)]{fuchs96} and the optimal POVM consists of the projectors onto an orthonormal eigenbasis of that operator. We describe the result here with a proof in terms of maximizing $G$. Note that this result is presented here for the sake of completeness and easy reference and we do not claim any contribution for this result.
\begin{lemma}
	\label{th:optPOVM}
	Given an interaction, an optimal POVM to attain maximum information gain consists of the eigenprojectors $\{E_\l\}$ onto the orthonormal eigenbasis $\{\ket{E_\l}\}$ that diagonalizes the Hermitian operator
	\begin{equation} \label{eq:gamma} \wt{\Gamma}_{xy} := p_x\rho_x - p_y\rho_y, \end{equation} 
	where $\rho_x$, as defined in Eq.~\eqref{eq:rho_x}, is the partial trace (over Alice's qubit) of the post-interaction state $\ket{X}$. The maximum achievable information gain is $\tr\left|\wt{\Gamma}_{xy}\right|$.
\end{lemma}
%\begin{lemma}
%\label{th:optPOVM}
%	Given an interaction $\{\ket{\xi_i},\ket{\zeta_j}\}$, over all POVMs $\{E_\l\}$, the value of information gain $G_{xy}$ is upper-bonded by $\sum_{i} \left|\gamma_i\right|$. The upper bound could be achieved by a POVM consisting of the eigenprojectors of \wt{\Gamma}_{xy}.
%\end{lemma}
\begin{proof} Given an interaction (i.e., the density operators $\rho_{x}, \rho_{y}$ get fixed), the associated $\wt{\Gamma}_{xy}$ being Hermitian is diagonalizable by an orthonormal eigenbasis $\{\ket{\gamma_i}\}$. Let the corresponding eigenvalues (all real) are $\{\gamma_i\}$. Then, over all POVMs $\{E_\l\}$,
	\begin{eqnarray*}  %\vspace{-15pt}
		G_{xy} &=& \sum\limits_\l \left|P_{\l x}p_x-P_{\l y}p_y\right| \\
%		\ &=& \sum\limits_\l \left| p_x~\tr\left(\rho_xE_\l\right) - p_y~\tr\left(\rho_yE_\l\right)\right|  \text{, using Eqs.}~\eqref{eq:PlxTr} \text{ and }~\eqref{eq:PlyTr} \\
		\ &=& \sum\limits_\l \left| p_x\tr\left(\rho_xE_\l\right) - p_y\tr\left(\rho_yE_\l\right)\right|  \text{, using Eqs.}~\mref{eq:PlxTr,eq:PlyTr} \\
		\ &=& \sum\limits_\l \left| \tr\left(\wt{\Gamma}_{xy} E_\l \right) \right| \text{, using Eq.}~\eqref{eq:gamma} \\
		\ &=& \sum\limits_\l \left|\sum\limits_{i} \gamma_i~\braopket{\gamma_i}{E_\l}{\gamma_i}  \right| \\
		%\ & & \text{[ where } \{\ket{\gamma_i}\}_i \text{ is a normalized eigenbasis for } \wt{\Gamma}_{xy} \atop %
		%\text{and } \{\gamma_i\}_i \text{ are the associated eigenvalues~] } %\hyperlink{propn_gamma}{\calcstepprf{prf:optPOVM}} 
		\\
		\ &\le& \sum\limits_\l \sum\limits_{i} \left|\gamma_i\right|~\braopket{\gamma_i}{ E_\l}{\gamma_i} \\ %\ && \text{[ equality occurs for mutually orthogonal }\{\ket{\gamma_i}\}_i ]\\
		\ &=& \sum\limits_{i} \left|\gamma_i\right|~\braopket{\gamma_i}{\sum\limits_\l E_\l}{\gamma_i} \\ 
		\ &=& \sum\limits_{i} \left|\gamma_i\right| 
		= \tr\left|\wt{\Gamma}_{xy}\right|. 
	\end{eqnarray*}
	The upper bound could be achieved by some POVM $\{E_\l\}$ consisting of the projectors onto an orthonormal eigenbasis of $\wt{\Gamma}_{xy}$.
%	Thus, $G_{xy}^{\star}$ is achieved by some POVM $\{E_\l\}$ that are projectors onto an orthonormal eigenbasis of $\wt{\Gamma}_{xy}$.
\end{proof}
\begin{remark}
	Since we consider equal prior probabilities here, analogous to Eq.~\eqref{eq:gamma}, we define 
	\begin{equation} \label{eq:gamma_eq} \Gamma_{xy} := \half\left(\rho_x - \rho_y\right) \end{equation}
	and use it throughout the rest of the paper.
\end{remark}

\begin{remark}
	Given an interaction, a POVM optimal for $G_{xy}$ may not necessarily be optimal for $I_{xy}$~\cite{fuchs97,fuchs96}. However, for equal prior, once the bound $\tr\left|\Gamma_{xy}\right|$  of $G_{xy}$ in Lemma~\ref{th:optPOVM} becomes equal to the upper bound $\Dp_{uv}^2-\Dm_{uv}^2$ of $G_{xy}$ in Eq.~\eqref{eq:boundG_xy}, the interaction is called optimal. In that case, the interaction-POVM pair also attains the upper bound~\eqref{eq:boundMI_xy} of $I_{xy}.$
\end{remark}
%\begin{remark}
%For an optimal interaction and the corresponding optimal POVM, the upper bound of $G_{xy}$ in Eq.~\eqref{eq:boundG_xy} should agree with the upper bound $\tr\left|\Gamma_{xy}\right| = \sum\limits_{\l} \left|\gamma_\l\right|$ of $G_{xy}$ in Lemma~\ref{th:optPOVM}.
%\end{remark}
%\clrA{Given an interaction, such a POVM that is optimal for $G_{xy}$, may not necessarily be optimal for $I_{xy}$~\cite{fuchs97}. Nevertheless, an interaction-POVM pair that is optimal for information gain by achieving the bounds~\mref{eq:boundG_xy,eq:boundG_uv} is also optimal for mutual information by achieving the bounds~\mref{eq:boundMI_xy,eq:boundMI_uv}.}

\subsection{Optimal interaction to maximize information gain ($G$): A generic form of optimal $\ket{\xi_i},\ket{\zeta_j}$ \label{optInteraction}} 
We use the following result for equal priors to find an expression of $\ket{\xi_i},\ket{\zeta_j}$ for optimal interaction.
\begin{lemma}\label{optGl}
	Optimality conditions for $G_{xy}$ ensure that each $G_{xy}^{\star}(\l)$ is equal to $G_{xy}^{\star}$ and the corresponding optimal value is given by 
	\begin{equation} \label{eq:G_opt} G_{xy}^{\star} = 2\concave{D_{uv}} = G_{xy}^{\star}(\l), ~\forall \l. \end{equation}
\end{lemma}
\begin{proof}
	For signal sent in $\xy$ basis, the optimal information gain, by Eq.~\eqref{eq:boundG_xy}, is	
	\begin{eqnarray*} 
	G_{xy}^{\star} &=& 2 \concave{D_{uv}}.
	\end{eqnarray*}
	By Eq.~\eqref{eq:boundGl_xy}, for measurement outcome $\l$ of Eve, %~\cite[Eq.~(20)]{fuchs97}
	\begin{eqnarray*} 
	G_{xy}^{\star}(\l) &=&  2\concavesq{d_{uv}(\l)}
	\end{eqnarray*}
	For optimality, the necessary and sufficient conditions in Proposition~\ref{condns:optG} must be satisfied. By~\cite[Eq.~(33)]{fuchs97}, this requires	
	\begin{eqnarray*} 
	d_{uv}(\l) = D_{uv}, ~\forall \l
	\end{eqnarray*}
	which ensures that the lemma is proved.	
\end{proof}
\begin{note}
Since we consider equal prior probabilities, we use the following working formula of $G_{xy}(\l)$ while we derive the general form of an optimal interaction, 
\begin{eqnarray} \label{eq:IG_lambda_eqprior}
G_{xy}(\l) &=& \left|Q_{x\l} - Q_{y\l}\right| = \frac{\left|P_{\l x} - P_{\l y}\right|}{P_{\l x} + P_{\l y}}.
\end{eqnarray}
\end{note}  
Here we describe an expression of $P_{\l x}, P_{\l y}$ in terms of $\ket{\xi_i},\ket{\zeta_j}$ and a POVM $\{E_\l\}$.
%By eq~\ref{eq:entangledExprXY}, 
\begin{theorem}
	Given the postinteraction sates~\eqref{eq:entangledExprXY}, and a POVM $\{E_\l\}_{\l\in\{0,1,2,3\}}$,
	\begin{eqnarray} \label{eq:Plx}
	P_{\l x} &=& (1-D_{xy})\braket{\xi_x}{E_\l}^2 + D_{xy}\braket{\zeta_x}{E_\l}^2, \nonumber \\
	P_{\l y} &=& (1-D_{xy})\braket{\xi_y}{E_\l}^2 + D_{xy}\braket{\zeta_y}{E_\l}^2.
	\end{eqnarray} 
\end{theorem}
\begin{proof}
	Using Eq.~\eqref{eq:entangledExprXY} in Eq.~\eqref{eq:rho_x}, we get,
	\begin{eqnarray} \label{expr:rho_x}
	\rho_x &=& \Trpart{A}\left(\ketbra{X}{X}\right) = (1-D_{xy})\wh{\xi}_x + D_{xy}\wh{\zeta}_x, 
	\end{eqnarray}
	where
	$$\wh{\xi}_x := \ketbra{\xi_x}{\xi_x}, \tab \wh{\zeta}_x := \ketbra{\zeta_x}{\zeta_x}.$$
	By Eq.~\eqref{eq:PlxTr},%~\eqref{eq:PlyTr},
	\begin{eqnarray*}
	P_{\l x}  &=& \Tr\left(\rho_xE_\l\right) \\
	\ &=& (1-D_{xy})\Tr\left(\wh{\xi}_xE_\l\right) + D_{xy}\Tr\left(\wh{\zeta}_xE_\l\right)\\
	\ &=& (1-D_{xy})\braket{\xi_x}{E_\l}^2 + D_{xy}\braket{\zeta_x}{E_\l}^2.
	\end{eqnarray*}
	Similarly, we can derive an expression for $P_{\l y}$. 
\end{proof}

We now have all the required pieces in place to derive the optimal interactions. First we gauge the difficulty of performing the derivation if we express the interaction vectors in canonical basis $\{\ket{\E_\l}\}$. We notice that the expressions~\mref{eq:Plx} of $P_{\l x}, P_{\l y}$ are dependent on the eigenvectors $\ket{E_\l}$ for which we don't have any easy formulation against an arbitrary interaction expressed in canonical basis. As we will see shortly, this blockage could be tackled if we express the interaction vectors $\ket{\xi_i},\ket{\zeta_j}$ in the orthonormal eigenbasis of the associated Hermitian $\Gamma_{xy}$.

Having understood this way of describing the interaction vectors, we start with a general form~\eqref{eq:xi-zeta-I} of $\ket{\xi_i},\ket{\zeta_j}$ expressed in the associated orthonormal eigenbasis $\{\ket{E_\l}\}$, while abiding by the orthogonality restriction~\eqref{rel:perp_opt}. Subsequently, we plug-in the expression~\eqref{eq:xi-zeta-I} of the interaction vectors into Eq.~\eqref{eq:Plx} to get the probabilities $P_{\l x}, P_{\l y}$. Then we substitute these probabilities into Eq.~\mref{eq:IG_lambda_eqprior} to get values of $G_{xy}(\l)$. Finally, comparing these values with their optimal counterparts in Eq.~\mref{eq:G_opt}, we derive the general form of an optimal interaction $\ket{\xi_i},\ket{\zeta_j}$ realized in eigenbasis $\{\ket{E_\l}\}$. 
	
This way of expressing interaction vectors not only helps us derive the optimal interactions, but, as we will see shortly, all the optimal interactions eventually lead to a unique expression.
%Now, consider an arbitrary orthonormal eigenbasis of $\Gamma_{xy}$ (which corresponds to an optimal POVM that leads to optimal $G_{xy}$). We derive the general form of $\ket{\xi_i},\ket{\zeta_j}$, described in that eigenbasis, for optimal interaction. 
\begin{theorem}
	Let $\{\ket{E_\l}\}$ be an orthonormal eigenbasis of $\Gamma_{xy}$ associated with arbitrary interaction vectors $\ket{\xi_i},\ket{\zeta_j}$ in Eq.~\eqref{eq:entangledExprXY} of the postinteraction states while abiding by the orthogonality restriction~\eqref{rel:perp_opt}. Then, for optimal interaction, the general form of $\ket{\xi_i},\ket{\zeta_j}$, described in that eigenbasis becomes 
	\begin{eqnarray} \label{eq:xi-zeta-I-final}
	\ket{\xi_x} &=& \Dp_{uv}~\ket{E_0} + \Dm_{uv}~\ket{E_1}, \nonumber\\
	\ket{\xi_y} &=& \Dm_{uv}~\ket{E_0} + \Dp_{uv}~\ket{E_1}, \nonumber\\
	\ket{\zeta_x} &=& \Dp_{uv}~\ket{E_2} + \Dm_{uv}~\ket{E_3}, \nonumber\\
	\ket{\zeta_y} &=& \Dm_{uv}~\ket{E_2} + \Dp_{uv}~\ket{E_3},
	\end{eqnarray}
	where $\Dp_{uv}, \Dm_{uv}$ are as defined in Eq.~\eqref{rel:DpDm}. 
\end{theorem}
\begin{proof}
	First we need to fix an orthonormal basis to describe $\ket{\xi_i},\ket{\zeta_j}$ following restriction~\eqref{rel:perp_opt}. For that purpose, there is no harm to choose the above eigenbasis to describe $\ket{\xi_i},\ket{\zeta_j}$. Orthogonality restriction~\eqref{rel:perp_opt} is automatically satisfied if we choose $\ket{\xi_i} \in \Span\{\ket{E_0}, \ket{E_1}\}$ and $\ket{\zeta_j} \in \Span\{\ket{E_2}, \ket{E_3}\}$. So the general form of $\ket{\xi_i},\ket{\zeta_j}$ becomes 
	\begin{eqnarray} \label{eq:xi-zeta-I}
	\ket{\xi_x} &=& \sqrt{\a}~\ket{E_0} + \sqrt{1-\a}~\ket{E_1}, \nonumber\\
	\ket{\xi_y} &=& \sqrt{\b}~\ket{E_0} + \sqrt{1-\b}~\ket{E_1}, \nonumber\\
	\ket{\zeta_x} &=& \sqrt{\mu}~\ket{E_2} + \sqrt{1-\mu}~\ket{E_3}, \nonumber\\
	\ket{\zeta_y} &=& \sqrt{\nu}~\ket{E_2} + \sqrt{1-\nu}~\ket{E_3}. 
	\end{eqnarray}
	Using this form of $\ket{\xi_i},\ket{\zeta_j}$ in Eq.~\eqref{eq:Plx}, we find values of $G_{xy}(\l)$ as shown in Table~\ref{tab:Gl}.
	
	\newcolumntype{R}{>{$}r<{$}}
	\newcolumntype{L}{>{$}l<{$}}

	\begin{table*}[htbp]
		%\begin{center} 
			\caption{Values of $P_{\l x}, P_{\l y}, G_{xy}(\l)$
				for the general form of $\ket{\xi_i},\ket{\zeta_j}$ as in Eq.~\eqref{eq:xi-zeta-I}.}  
			\label{tab:Gl}
	\renewcommand{\arraystretch}{1.8}
	\begin{tabular*}{\linewidth}{R@{\extracolsep{\fill}}LLL}
		\toprule \rule{0pt}{4.5ex} 
		\l & P_{\l x} & P_{\l y} & G_{xy}(\l) = \dfrac{\left|P_{\l x} - P_{\l y}\right|}{P_{\l x} + P_{\l y}} \\[1.5ex]
		\hline   
		0 & \left(1-D_{xy}\right)\braket{\xi_x}{E_0}^2=\left(1-D_{xy}\right)\a & \left(1-D_{xy}\right)\braket{\xi_y}{E_0}^2=\left(1-D_{xy}\right)\b & \nicefrac{\left|\a-\b\right|}{\left(\a+\b\right)}  \\  
		%\hline 
		1 & \left(1-D_{xy}\right)\braket{\xi_x}{E_1}^2=\left(1-D_{xy}\right)(1-\a) & \left(1-D_{xy}\right)\braket{\xi_y}{E_1}^2=\left(1-D_{xy}\right)(1-\b) & \nicefrac{\left|\b-\a\right|}{\left(1-\a+1-\b\right)}  \\ 
		%\hline 
		2 & D_{xy}\braket{\zeta_x}{E_2}^2=D_{xy}\m & D_{xy}\braket{\zeta_y}{E_2}^2=D_{xy}\n & \nicefrac{\left|\m-\n\right|}{\left(\m+\n\right)}  \\  
		%\hline 
		3 & D_{xy}\braket{\zeta_x}{E_3}^2=D_{xy}(1-\m) & D_{xy}\braket{\zeta_y}{E_3}^2=D_{xy}(1-\n) & \nicefrac{\left|\n-\m\right|}{\left(1-\m+1-\n\right)} \\ 
		\hline \hline  %\bottomrule
	\end{tabular*}	
%\end{center}	
\end{table*}		
	\renewcommand{\arraystretch}{1}

	By Lemma~\ref{optGl} , for optimal $G_{xy}$, the values of $G_{xy}(\l)$ are all equal. 
	Equating $G_{xy}(0), G_{xy}(1)$ in Table~\ref{tab:Gl}, we get,
	\begin{eqnarray*}
		\a+\b=1,  &\quad & G_{xy}(0)=G_{xy}(1)=\left|2\a-1\right|.
	\end{eqnarray*}
	Similarly, equating $G_{xy}(2), G_{xy}(3)$ in Table~\ref{tab:Gl}, we get,
	\begin{eqnarray*}
		\m+\n=1,  &\quad & G_{xy}(2)=G_{xy}(3)=\left|2\m-1\right|.
	\end{eqnarray*}
	Together, equating $G_{xy}(0), G_{xy}(2)$, we get,
	\begin{eqnarray} \label{eq:opt_alpha_rel}
	\m=\a, &~~~ & \n=\b=1-\a.
	\end{eqnarray}
	Thus,
	\begin{eqnarray*}
		G_{xy}^{\star}(0) = \Dp_{uv}^2-\Dm_{uv}^2 = 2\Dp_{uv}^2 -1  = \left|2\a-1\right|
	\end{eqnarray*}
	gives rise to
	\begin{eqnarray} \label{eq:opt_alpha_rt}
	\sqrt{\a} = \Dp_{uv}, &~~~& \sqrt{1-\a} =\Dm_{uv}.
	\end{eqnarray}
	Using Eqs.~\mref{eq:opt_alpha_rt,eq:opt_alpha_rel} in Eq.~\eqref{eq:xi-zeta-I}, we get a generic form for optimal $\ket{\xi_i},\ket{\zeta_j}$ as in Eq.~\eqref{eq:xi-zeta-I-final}.  
\end{proof} 
Analogous to Eq.~\eqref{eq:xi-zeta-I-final}, a set of optimal interaction vectors exist in $\uv$ basis.

The most interesting thing with the expression~\eqref{eq:xi-zeta-I-final} of the optimal interaction vectors is that it has a unique form capturing all the  optimal interactions while realized in the orthonormal eigenbasis of the associated $\Gamma_{xy}.$

%Note that a complete description of an optimal interaction by Eve could be realized, in a unique form, by  combing optimal interaction vectors $\ket{\xi_i},\ket{\zeta_j}$ of Eq.~\eqref{eq:xi-zeta-I-final} with the postinteraction sates $\ket{X}, \ket{Y}$ of Eq.~\eqref{eq:entangledExprXY}. 
 \begin{remark}
 	An optimal interaction for equal error rates could be described by an expression analogous to  Eq.~\eqref{eq:xi-zeta-I-final} while $\Dp_{uv}, \Dm_{uv}$ are replaced by $\Dp,\Dm$ respectively.
 \end{remark}
\begin{remark}
	We can rewrite Eq.~\eqref{eq:xi-zeta-I-final} as below.
	\begin{eqnarray} \label{eq:xi-zeta-I-final:D}
	\ket{\xi_x} &=& \sqrt{1-D_{uv}}~\tE{0}+\sqrt{D_{uv}}~\tE{1},\nonumber\\
	\ket{\xi_y} &=& \sqrt{1-D_{uv}}~\tE{0}-\sqrt{D_{uv}}~\tE{1}, \nonumber\\
	\ket{\zeta_x} &=& \sqrt{1-D_{uv}}~\tE{2}+\sqrt{D_{uv}}~\tE{3}, \nonumber\\
	\ket{\zeta_y} &=& \sqrt{1-D_{uv}}~\tE{2}-\sqrt{D_{uv}}~\tE{3},  
	\end{eqnarray}
	where
	\begin{eqnarray} \label{eq:basis_E-tilde}
	\tE{0} = \invsqrttwo\left(\ket{E_0} + \ket{E_1}\right), &~&
	\tE{1} = \invsqrttwo\left(\ket{E_0} - \ket{E_1}\right), \nonumber\\
	\tE{2} = \invsqrttwo\left(\ket{E_2} + \ket{E_3}\right), &~&
	\tE{3} = \invsqrttwo\left(\ket{E_2} - \ket{E_3}\right).  \nonumber\\
	\ && \
	\end{eqnarray}
	is another orthonormal basis (called, Bell basis), written in terms of an optimal eigenbasis $\{E_\l\}$. Clearly, these form to describe $\ket{\xi_i},\ket{\zeta_j}$ is analogous to Eqs.~(51) and (50), respectively, as in \cite{fuchs97}. 
\end{remark}
%\clrA{Once the unique expression~\eqref{eq:xi-zeta-I-final} of the optimal interaction vectors is derived, the natural questions that arise are the following: how does it lead to specific instances like~\mref{eq:opt_xi-zeta} or~\mref{eq:choice:xi_zeta:opt}? What are the other possible instantiations?} 
As we will see in the next subsections, expression~\eqref{eq:xi-zeta-I-final} hides a family of optimal interactions via the eigenbasis $\{\ket{E_\l}\}$ -- we can unfold it once we identify the associated optimal POVMs. Since an optimal POVM corresponds to some $\Gamma_{xy}$, we need to find the expression of $\Gamma_{xy}$ realizing interactions~\eqref{eq:xi-zeta-I-final}.

%Since optimal $G_{xy}$ is achieved by expression~\eqref{eq:xi-zeta-I-final} of $\ket{\xi_i}, \ket{\zeta_j}$, by Lemma~\ref{th:optPOVM}, the corresponding $\Gamma_{xy}$ should be diagonalized in its eigenbasis. It is confirmed by the following theorem, while we figure out the eigenvalues $\{\gamma_i\}$ of $\Gamma_{xy}$ for an optimal interaction.
%Since the expression~\eqref{eq:xi-zeta-I-final} of $\ket{\xi_i}, \ket{\zeta_j}$ corresponds to optimal $G$, therefore, by Lemma~\ref{th:optPOVM}, $\Gamma_{xy}$ should become a diagonal matrix for this form of $\ket{\xi_i}, \ket{\zeta_j}$ in its eigenbasis. We confirm this, while we figure out the eigenvalues, as shown below.
\begin{theorem}
	For an optimal interaction~\mref{eq:xi-zeta-I-final,eq:entangledExprXY}, and its optimal POVM $\{E_\l\}$,
	\begin{eqnarray} \label{eq:Gxy:inEigenbasis}	
	\Scale[0.98]{
	\Gamma_{xy} \!=\! \half(\Dp_{uv}^2\!-\!\Dm_{uv}^2)[(1\!-\!D_{xy})(\EE_{00}\!-\!\EE_{11}) \!+\! D_{xy}(\EE_{22}\!-\!\EE_{33})],} \nonumber\\
	\end{eqnarray}
%	\begin{eqnarray} \label{eq:Gxy:inEigenbasis}
%		\Gamma_{xy} &=& \half(\Dp_{uv}^2-\Dm_{uv}^2)\left[\left(1-D_{xy}\right)\left(\EE_{00}-\EE_{11}\right) \right. \nonumber\\
%		\ && \left. + D_{xy}\left(\EE_{22}-\EE_{33}\right)\right], 
%	\end{eqnarray}
	where  \qquad\qquad\qquad $\EE_{ij} := \ketbra{E_i}{E_j}.$      
%	$$\EE_{ij} := \ketbra{E_i}{E_j}.$$
\end{theorem} 
\begin{proof}
	By Eq.~\eqref{expr:rho_x} and its analogue for signal $y$, 
	\begin{eqnarray*}
		2\Gamma_{xy} &=& \rho_x-\rho_y = (1-D_{xy})\left(\wh{\xi}_x-\wh{\xi}_y\right) + D_{xy}\left(\wh{\zeta}_x-\wh{\zeta}_y\right).
	\end{eqnarray*}
	Using expressions of $\ket{\xi_i}, \ket{\zeta_j}$ in Eq.~\eqref{eq:xi-zeta-I-final}, we get,
	\begin{eqnarray*}
		\wh{\xi}_x &=& \Dp_{uv}^2~\EE_{00} +\Dm_{uv}^2~\EE_{11} + 2\Dp_{uv}\Dm_{uv}\left(\EE_{01}+\EE_{10}\right), \nonumber \\
		\wh{\xi}_y &=& \Dm_{uv}^2~\EE_{00} +\Dp_{uv}^2~\EE_{11} + 2\Dp_{uv}\Dm_{uv}\left(\EE_{01}+\EE_{10}\right), \nonumber \\
		\wh{\zeta}_x &=& \Dp_{uv}^2~\EE_{22} +\Dm_{uv}^2~\EE_{33} + 2\Dp_{uv}\Dm_{uv}\left(\EE_{23}+\EE_{32}\right), \nonumber \\
		\wh{\zeta}_y &=& \Dm_{uv}^2~\EE_{22} +\Dp_{uv}^2~\EE_{33} + 2\Dp_{uv}\Dm_{uv}\left(\EE_{23}+\EE_{32}\right). \nonumber 
	\end{eqnarray*}
	which leads to the desired form of $\Gamma_{xy}$.
\end{proof}
\begin{remark} 
	Clearly, $\Gamma_{xy}$ in~\eqref{eq:Gxy:inEigenbasis} is diagonalized by its eigenbasis while its eigenvalues are	       
	\begin{eqnarray} \label{eq:Eigenval}
	\gamma_0 = \half\left(\Dp_{uv}^2-\Dm_{uv}^2\right)\left(1-D_{xy}\right), &~~& \gamma_1 = -\gamma_0,\nonumber \\
	\gamma_2 = \half\left(\Dp_{uv}^2-\Dm_{uv}^2\right)D_{xy}, &~~& \gamma_3 = -\gamma_2. 
	 \end{eqnarray}   
\end{remark}
It is worth to note here that, for interactions~\eqref{eq:xi-zeta-I-final}, the optimal value $\Dp_{uv}^2-\Dm_{uv}^2$ of $G_{xy}$ in Eq.~\eqref{eq:boundG_xy} agrees with the upper bound $\sum\limits_{\l} \left|\gamma_\l\right|$ of $G_{xy}$ in Lemma~\ref{th:optPOVM}.

%Before discussing the class of optimal POVMs corresponding to our generic description of an optimal interaction~\eqref{eq:xi-zeta-I-final}, we consider a special case - we derive the POVM for the optimal interaction~\eqref{eq:choice:xi_zeta:opt} for equal error rates in the next subsection.

As a first step towards identifying the optimal POVMs hiding behind the interactions~\eqref{eq:xi-zeta-I-final}, we connect them with the known instances~\mref{eq:opt_xi-zeta,eq:choice:xi_zeta:opt}. 
	
In Eq.~\eqref{eq:xi-zeta-I-final}, a mere replacement of the eigenbasis $\{\ket{E_\l}\}$ by the canonical basis $\{\ket{\E_\l}\}$ leads to the interaction~\mref{eq:opt_xi-zeta}, except for a trivial permutation of the basis vectors. The corresponding $\Gamma_{xy}$ in Eq.~\eqref{eq:Gxy:inEigenbasis} becomes
	\begin{eqnarray*} 	
	\Scale[0.98]{
		\Gamma_{xy} \!=\! \half(\Dp_{uv}^2\!-\!\Dm_{uv}^2)[(1\!-\!D_{xy})(\E_{00}\!-\!\E_{11}) \!+\! D_{xy}(\E_{22}\!-\!\E_{33})],} 
	\end{eqnarray*}
where $\E_{ii}$	stands for $\ketbra{\E_{i}}{\E_{i}}$. Clearly, the canonical basis $\{\ket{\E_\l}\}$ diagonalizes this $\Gamma_{xy}$ and therefore constitutes the optimal POVM~\eqref{eq:optEl} for the said interaction. The diagonalization worked due to orthonormality of the canonical basis. 

Connecting Eq.~\eqref{eq:xi-zeta-I-final} to the instance~\mref{eq:choice:xi_zeta:opt} is done in the following subsection.
%\clrA{Connecting Eq.~\eqref{eq:xi-zeta-I-final} to the known instance~\mref{eq:opt_xi-zeta} along with its optimal POVM~\eqref{eq:optEl} remained easy. But it's not so direct to connect Eq.~\eqref{eq:xi-zeta-I-final} to the instance~\mref{eq:choice:xi_zeta:opt}. The related workaround is done in the following subsection.}

\subsection{Optimal POVM for the specific interaction for equal error rates ($D_{xy}=D_{uv}=D$) by Fuchs \etal. \cite{fuchs97} \label{optPOVM_eqD}}
For equal error rates, i.e., $D_{xy}=D_{uv}=D$, \cite{fuchs97} described a choice of $\ket{\xi_i}, \ket{\zeta_j}$, that optimizes $I$ (and therefore $G$). For this optimal $\ket{\xi_i}, \ket{\zeta_j}$ as described in Eq.~\eqref{eq:choice:xi_zeta:opt}, we now derive the optimal POVM that was not shown in~\cite{fuchs97}. To see how Eq.~\eqref{eq:xi-zeta-I-final} produces Eq.~\eqref{eq:choice:xi_zeta:opt}, one can simply compare the respective interaction vectors in each of these equations. The comparison gives rise to a set of vectors $\{\ket{E_\l}\}$, which constitutes an optimal POVM if it diagonalizes the corresponding $\Gamma_{xy}$.
\begin{theorem}\label{th:POVM-Fuchs}
	Consider a canonical basis for Eve as given in Eq.~\eqref{eq:basis_cano_Eve}. 
	For the optimal interaction~\eqref{eq:choice:xi_zeta:opt}, the optimal POVM $\{E_\l\}$ could be given by
	$$\El = \ketbra{\El}{\El},$$
	where 
	\begin{eqnarray} \label{eq:optimal_POVM-vec:eqErr}
	\ket{E_0} = \Dp\ket{\E_0}-\Dm\ket{\E_1}, & \ & \ket{E_1} = \Dm\ket{\E_0}+\Dp\ket{\E_1}, \nonumber\\
	\ket{E_2} = \Dp\ket{\E_2}-\Dm\ket{\E_3}, & \ & \ket{E_3} = \Dm\ket{\E_2}+\Dp\ket{\E_3}.\nonumber\\
	\end{eqnarray} 
\end{theorem}
\begin{proof} Comparing a special form of $\ket{\xi_x}, \ket{\xi_y}$ given by Eq.~\eqref{eq:choice:xi_zeta:opt} and the general form of $\ket{\xi_x}, \ket{\xi_y}$ described in Eq.~\eqref{eq:xi-zeta-I-final} but for equal error rates, we get
	\begin{eqnarray*} 
	\Dp~\ket{E_0} + \Dm~\ket{E_1} &=& \ket{\E_0}, \\ %\label{eq:bases1}
	\Dm~\ket{E_0} + \Dp~\ket{E_1} &=& 2\Dp\Dm~\ket{\E_0}+\left(\Dp^2-\Dm^2\right)\ket{\E_1}. %\nonumber\\	\ && \label{eq:bases2}
	\end{eqnarray*}
	Solving for $\ket{E_0}$ and $\ket{E_1}$, one should arrive at the first two expressions of Eq.~\eqref{eq:optimal_POVM-vec:eqErr}.
%	\begin{eqnarray*} 
%	\ket{E_0} & = & \Dp\ket{\E_0} - \Dm\ket{\E_1},\\
%	\ket{E_1} & = & \Dm\ket{\E_0} + \Dp\ket{\E_1}.
%	\end{eqnarray*}
	The remaining two expressions in Eq.~\eqref{eq:optimal_POVM-vec:eqErr} could be derived by comparing the expressions of $\ket{\zeta_x}, \ket{\zeta_y}$ in Eqs.~\mref{eq:choice:xi_zeta:opt,eq:xi-zeta-I-final}.  That these vectors diagonalize the associated $\Gamma_{xy}$ is proven in the following theorem.
%	we can establish that
%	\begin{eqnarray*} 
%	\ket{E_2} & = & \Dp\ket{\E_2} - \Dm\ket{\E_3},\\
%	\ket{E_3} & = & \Dm\ket{\E_2} + \Dp\ket{\E_3}.
%	\end{eqnarray*}
\end{proof}
While realized in the basis~\eqref{eq:optimal_POVM-vec:eqErr}, the expression~\eqref{eq:Gxy:inEigenbasis} of $\Gamma_{xy}$ gets diagonalized by the same basis vectors -- the diagonalization works because the vectors under consideration are orthonormal. Formally speaking,
\begin{theorem}\label{th:orthonormality}
For $\Gamma_{xy}$ described in Eq.~\eqref{eq:Gxy:inEigenbasis} but for equal error rates and realized in the basis~\eqref{eq:optimal_POVM-vec:eqErr}, $$\Gamma_{xy}\ket{E_0} = \gamma_0\ket{E_0},$$
where $\gamma_0 = \concave{D}(1-D) = \half\left(\Dp^2-\Dm^2\right)(1-D)$ for equal error rates.
\end{theorem}  
\begin{proof}
	For equal error rates, Eq.~\eqref{eq:Gxy:inEigenbasis} becomes
	%	\begin{eqnarray*} 
	%		\Gamma_{xy} &=& \left(\Dp^2-\Dm^2\right)\left[\left(1-D\right)\left(\EE_{00}-\EE_{11}\right) \right. \nonumber\\
	%		\ && \left. + D\left(\EE_{22}-\EE_{33}\right)\right].
	%	\end{eqnarray*}
	\begin{eqnarray*} 
		\Gamma_{xy} = \half(\Dp^2-\Dm^2)\left[\left(1-D\right)\left(\EE_{00}-\EE_{11}\right) + D\left(\EE_{22}-\EE_{33}\right)\right].
	\end{eqnarray*}
	If we realize this expression in the basis~\eqref{eq:optimal_POVM-vec:eqErr}, then
	\begin{eqnarray*}
		\Gamma_{xy}\ket{E_0} %&=& \left(\Dp^2-\Dm^2\right)\left(1-D\right)\EE_{00}\ket{E_0} \\
		&=& \half\left(\Dp^2-\Dm^2\right)\left(1-D\right)\ket{E_0} = \gamma_0\ket{E_0},
	\end{eqnarray*}
	so far $\{\ket{E_i}\}$ are orthonormal, which indeed is true for Eq.~\eqref{eq:optimal_POVM-vec:eqErr}.
	This completes the proof for $\l=0$. 
\end{proof}
\begin{remark}
	One can calculate and check that the eigenvalues of $\Gamma_{xy}$ described in eigenbasis~\eqref{eq:optimal_POVM-vec:eqErr} match those as in Eq.~\eqref{eq:Eigenval} calculated for the generic form of optimal $\ket{\xi_i}, \ket{\zeta_j}$ described by Eq.~\eqref{eq:xi-zeta-I-final} but for equal error rates.
\end{remark}
%This section has provided us with an insight of a possible generic form of the optimal POVM. We continue with the task to find %a generic form of the 
%the class of optimal POVMs in the next section. 
%--commented out--%	This subsection has provided us with a specific rotation of the canonical basis as the optimal POVM for a specific optimal interaction. Since any rotation of the canonical basis retains orthonormality, in the following subsection, we attempt to correspond every \clrA{rotation of the canonical} basis as an optimal POVM against an optimal interaction described in the same rotated basis. 

Connecting interactions~\eqref{eq:xi-zeta-I-final} to the known instances~\mref{eq:opt_xi-zeta,eq:choice:xi_zeta:opt} helped develop the intuition for finding the family of optimal POVMs hidden behind interactions~\eqref{eq:xi-zeta-I-final}. Now we are in a position to pinpoint the bases for which $\Gamma_{xy}$ in Eq.~\eqref{eq:Gxy:inEigenbasis} gets diagonalized.

%\subsection{Optimal POVM to maximize information gain ($G$): A generic form \label{optPOVM:form}}
\subsection{Optimal POVM corresponding to an optimal interaction in its generic form\label{optPOVM:form}}

In Theorem~\ref{th:POVM-Fuchs}, it was shown that a specific rotation of the canonical basis $\{\ket{\E_\l}\}$ yields an eigenbasis $\{\ket{E_\l}\}$ of $\Gamma_{xy}$ that corresponds to the optimal interaction~\eqref{eq:xi-zeta-I-final}. Now, we will show that not only the above specific rotation, but any rotation represented by an orthogonal linear transformation of the canonical basis $\{\ket{\E_\l}\}$ yields an eigenbasis of $\Gamma_{xy}$ in Eq.~\eqref{eq:Gxy:inEigenbasis}. %corresponding to the optimal interaction~\eqref{eq:xi-zeta-I-final}. In the next subsection, we proceed one step further to show that such a transformation preserves the optimality of both information gain and mutual information.
%\cref{eq:tbasis_canonical,eq:Gxy:inEigenbasis}

Given an interaction, an optimal POVM corresponds to an orthonormal basis $\{\ket{E_\l}\}$ that diagonalizes $\Gamma_{xy}$ associated with that interaction. For interactions~\eqref{eq:xi-zeta-I-final}, the associated $\Gamma_{xy}$ is given by Eq.~\eqref{eq:Gxy:inEigenbasis}. Our task is to identify the bases, each of which diagonalizes the corresponding $\Gamma_{xy}$ realized in the same 
basis. First we observe that, for a set $\{\ket{E_\l}\}$ of vectors, and for projectors $\EE_{ii}$ in Eq.~\eqref{eq:Gxy:inEigenbasis}, 
\vspace{-1em}\[ \EE_{ii}\ket{E_j} = \delta_{ij}\ket{E_i}, \text{~~with~} \delta_{ij} \text{~the Kronecker delta,~}\] 
%	\vspace{0em}\[ \EE_{ii}\ket{E_j} = 
%	\begin{cases}	
%	\ket{E_i} \text{, if } i=j, \\
%	\Ovec \text{, if } i\neq j.
%	\end{cases}
%	\] 
	if and only if the vectors $\ket{E_\l}$ are orthonormal.
In such case, it is guaranteed that the $\Gamma_{xy}$ in Eq.~\eqref{eq:Gxy:inEigenbasis} gets diagonalized by the basis $\{\ket{E_\l}\}$, because
	\begin{eqnarray*} \Gamma_{xy}\ket{E_\l} = \gamma_\l\ket{E_\l}, ~\forall \l, \end{eqnarray*}
where the values of $\gamma_\l$ coincides with those in Eq.~\eqref{eq:Eigenval}.

Since any rotation $\rot$ of the canonical basis $\{\ket{\E_\l}\}$ produces an orthonormal basis, it diagonalizes the $\Gamma_{xy}$ of Eq.~\eqref{eq:Gxy:inEigenbasis} realized in the same rotated basis. Hence, each of these rotated basis constitutes the optimal POVM for an interaction~\eqref{eq:xi-zeta-I-final} realized in the same rotated basis.

To express the above idea mathematically, we introduce the notations below:
\begin{eqnarray*}
	\vegn := (\ket{E_0},\ket{E_1},\ket{E_2},\ket{E_3})^T, &~~& \vcan := (\ket{\E_0},\ket{\E_1},\ket{\E_2},\ket{\E_3})^T,
\end{eqnarray*} 
and state the following result.
%\begin{theorem}
%	Consider a canonical basis $\{\ket{\E_\l}\}$ for Eve as given in Eq.~\mref{eq:basis_cano_Eve}. Let 
%	\begin{eqnarray} \label{eq:optimal_POVM-vec}
%	(\ket{E_i}) = \rot(\ket{\E_j})
%	\end{eqnarray}
%	where $(\ket{E_i})$ corresponds to the column vector $(\ket{E_0},\ket{E_1},\ket{E_2},\ket{E_3})^{T}$ (similar meaning for $(\ket{\E_j})$) and $\rot$ is an orthogonal matrix. %$\colvec{\ket{E_0}}{\ket{E_1}}{\ket{E_2}}{\ket{E_3}}$.
%	Then $\{E_\l\}_{\l\in\{0,1,2,3\}}$ forms an orthonormal eigenbasis for $\Gamma_{xy}$.
%\end{theorem}
\begin{theorem}
	Any orthogonal rotation $\rot$ of the canonical basis $\{\ket{\E_\l}\}$, which is realized by 
	\begin{eqnarray} \label{eq:optimal_POVM-vec}
	\vegn = \rot\vcan,
	\end{eqnarray}
	works as an orthonormal eigenbasis of $\Gamma_{xy}$ in Eq.~\eqref{eq:Gxy:inEigenbasis} while the optimal interaction~\eqref{eq:xi-zeta-I-final} is realized in the same rotated basis. To be precise, the rotated basis diagonalizes $\Gamma_{xy}$ in this case. 
\end{theorem}

Note that in Eq.~\eqref{eq:optimal_POVM-vec}, orthonormality of the eigenbasis is preserved only when $\rot$ is an orthogonal matrix, and not any arbitrary linear transformation. 
Thus, for any orthogonal rotation $\rot$, the orthonormal basis $\{\ket{E_\l}\}$ realized by Eq.~\eqref{eq:optimal_POVM-vec} corresponds to an optimal POVM while the optimal interaction~\eqref{eq:xi-zeta-I-final} is also realized in the same rotated basis. To see how these optimal POVMs help Eq.~\eqref{eq:xi-zeta-I-final} to generate the family of optimal interactions, we let the coefficient matrix of Eq.~\eqref{eq:xi-zeta-I-final} as $\matd$. Then, an optimal POVM due to $\vegn = \rot\vcan$ leads to an instantiation $\matd\vegn=\matd\rot\vcan$ of the optimal interaction. Thus, Eqs.~\eqref{eq:xi-zeta-I-final} and~\eqref{eq:optimal_POVM-vec} establishes a one-to-one correspondence between an optimal interaction $\matd\rot\vcan$ and the optimal POVM realizing $\rot\vcan$. Fixing a rotation $\rot$ provides an instance of such pairs $\left(\matd\rot\vcan, ~\rot\vcan\right)$.
Here we create a subclass of such orthogonal rotation. 
\begin{example}\label{ex1}
	Consider a canonical basis $\{\ket{\E_\l}\}$ for Eve as given in Eq.~\eqref{eq:basis_cano_Eve}. Let 
	\begin{eqnarray} \label{eq:optimal_POVM-vec_exmpl}
	\ket{E_0} = \sqrt{a}\ket{\E_0}-\sqrt{1-a}\ket{\E_1}, \nonumber \\
  \ket{E_1} = \sqrt{1-a}\ket{\E_0}+\sqrt{a}\ket{\E_1}, \nonumber \\
	\ket{E_2} = \sqrt{a}\ket{\E_2}-\sqrt{1-a}\ket{\E_3}, \nonumber \\
 \ket{E_3} = \sqrt{1-a}\ket{\E_2}+\sqrt{a}\ket{\E_3}.
	\end{eqnarray}
	Since the coefficient matrix is orthogonal, $\{\ket{E_\l}\}_{\l\in\{0,1,2,3\}}$ forms an orthonormal eigenbasis for $\Gamma_{xy}$. %Thus Eq.~\eqref{eq:optimal_POVM-vec_exmpl} corresponds to an optimal POVM.
\end{example}

\subsection{Achieving both optimal information gain (G) and optimal mutual information (I) \label{opttimality:check}}
Identification of the optimal POVMs helped to unfold the family of optimal interactions that was hiding behind the unique expression~\eqref{eq:xi-zeta-I-final}. Equations~\eqref{eq:xi-zeta-I-final} and~\eqref{eq:optimal_POVM-vec}, when considered together, provides a family of interaction-POVM pairs $\left(\matd\rot\vcan, ~\rot\vcan\right)$. Such a pair along with their counterpart in $\uv$ basis, by virtue of our construction of optimal interactions~\eqref{eq:xi-zeta-I-final}, should lead to optimal information gain by achieving the bounds in Eqs.~\mref{eq:boundG_xy,eq:boundG_uv}, which in turn lead to optimal mutual information by achieving the bounds in Eqs.~\mref{eq:boundMI_xy,eq:boundMI_uv}. However, for completeness, we show that such a pair satisfies the necessary and sufficient conditions given by Proposition~\ref{condns:optG} and therefore leads to optimality. Following this process, as an indicator for optimality, we establish in Lemma~\ref{rel:epsa_l:egnval} an additional result regarding the sign of $\gamma_{\l}$. 
	
Proof of the following theorem works on the unique expression $\matd\vegn$ of optimal interaction vectors rather than working on its representative pairs $\left(\matd\rot\vcan, ~\rot\vcan\right)$.
The initial task is to find an expression for $\ket{\xi_u},\ket{\xi_v},\ket{\zeta_u},\ket{\zeta_v}$ corresponding to Eq.~\eqref{eq:xi-zeta-I-final}. For this, we use Eqs.~\mref{eq:xi-zeta-I-final:D,eq:basis_E-tilde} in Eqs.~\mref{eq:xiuxy,eq:xivxy}, to derive the following intermediate result.
\begin{lemma} 
For achieving the optimal information gain, we must have
	\begin{eqnarray} \label{eq:xi-zeta-u-final:D}
	\ket{\xi_u} &=& \sqrt{1-D_{xy}}~\tE{0}+\sqrt{D_{xy}}~\tE{2}\nonumber\\
	\ket{\xi_v} &=& \sqrt{1-D_{xy}}~\tE{0}-\sqrt{D_{xy}}~\tE{2} \nonumber\\
	\ket{\zeta_u} &=& \sqrt{1-D_{xy}}~\tE{1}-\sqrt{D_{xy}}~\tE{3} \nonumber\\
	\ket{\zeta_v} &=& \sqrt{1-D_{xy}}~\tE{1}+\sqrt{D_{xy}}~\tE{3}  
	\end{eqnarray}
	where the basis $\{\tE{\l}\}$ is as described in Eq.~\eqref{eq:basis_E-tilde}.
\end{lemma}
\begin{remark}
To get expressions of $\ket{\xi_i},\ket{\zeta_j}$ in $\uv$ basis symmetric to those in $\xy$ basis, e.g., like \cite[Eq.~(52)]{fuchs97}, one must consider the canonical basis in the order $\ket{\E_0}=\txx, \ket{\E_1}=\tyy, \ket{\E_2}=\txy, \ket{\E_3}=\tyx$, compatible with \cite{fuchs97}.
\end{remark}

\begin{figure*}[htbp]
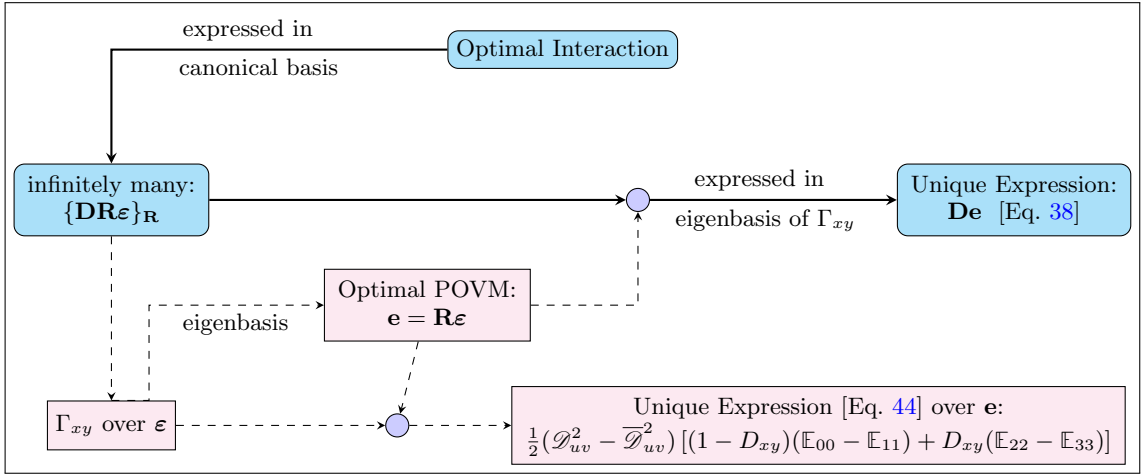

	\caption{Optimal interaction: infinitely many in canonical basis, while unique in eigenbasis} \label{fig:many2unq}
	\centering \boxed{\tikzGist}
\end{figure*}

\begin{theorem} The interaction given by Eqs.~\mref{eq:xi-zeta-I-final,eq:entangledExprXY} and a POVM corresponding to the eigenbasis given by Eq.~\eqref{eq:optimal_POVM-vec} satisfy the necessary and sufficient conditions given by Proposition~\ref{condns:optG} and therefore attain optimal information gain and optimal mutual information.
\end{theorem}
\begin{proof} From Eq.~\eqref{eq:Ulu}, we have
	\begin{eqnarray*}
		\ket{U_{\l u}} &=& B_{u}\tensorprod\sqrt{E_\l}~ \ket{U} = B_{u}\tensorprod{E_\l}~ \ket{U} \\
		%\ &=& B_{u}\tensorprod{E_1}~\left(\entangledExpr{uv}{u}{u}{v}{u}\right), \\
		%\ && \text{by Eq.~}\eqref{eq:entangledExprUV} \\
		\ &=& \sqrt{1-D_{uv}}~\left(B_{u}\ket{u}\right)\tensorprod\left(E_\l\ket{\xi_{u}}\right) \\
		\ && + \sqrt{D_{uv}}~\left(B_{u}\ket{v}\right)\tensorprod\left(E_\l\ket{\zeta_{v}}\right), \text{by Eq.~}\eqref{eq:entangledExprUV}. 
	\end{eqnarray*}
	Since $B_{u}\ket{u} = \ket{u}, B_{u}\ket{v} = \Ovec$, and $E_\l\ket{\xi_{u}} = \braket{E_\l}{\xi_u}\ket{E_\l}$, we get,
	\begin{eqnarray*}
		\ket{U_{\l u}} &=&  \sqrt{1-D_{uv}}~\braket{E_\l}{\xi_u}\ket{u}\ket{E_\l}.
	\end{eqnarray*}
	Similarly,
	\begin{eqnarray*}
		\ket{V_{\l u}} &=& 
		%		B_{u}\tensorprod\sqrt{E_1}~ \ket{V} = B_{u}\tensorprod{E_1}~ \ket{V} \\
		%		\ &=& B_{u}\tensorprod{E_1}~\left(\entangledExpr{uv}{v}{v}{u}{v}\right), \text{by Eq.~}\eqref{eq:entangledExprUV} \\ 	\ &=& 
		\sqrt{D_{uv}}~\braket{E_\l}{\zeta_v}\ket{u}\ket{E_\l}.
	\end{eqnarray*}
	Here, we want equality in magnitude between $\braket{E_\l}{\xi_u}$ and $\braket{E_\l}{\zeta_v}$. 	
	Now,  by Eq.~\eqref{eq:xi-zeta-u-final:D}, $\braket{E_\l}{\xi_u}$ takes values $$\invsqrttwo\sqrt{1-D_{xy}},\invsqrttwo\sqrt{1-D_{xy}},\invsqrttwo\sqrt{D_{xy}},\invsqrttwo\sqrt{D_{xy}}~;$$ whereas, $\braket{E_\l}{\zeta_v}$ takes values  $$\invsqrttwo\sqrt{1-D_{xy}},-\invsqrttwo\sqrt{1-D_{xy}},\invsqrttwo\sqrt{D_{xy}},-\invsqrttwo\sqrt{D_{xy}}~,$$ respectively for $\l=0,1,2,3$.
	Therefore, 
	\begin{eqnarray*} 
		\ket{V_{\l u}} &=& \varepsilon_\l~\factor{uv}~ ~\ket{U_{\l u}},
	\end{eqnarray*}
	where,
	\begin{equation}
	\label{eq:epsa_val}
	\varepsilon_0 = +1, \indent \varepsilon_1 = -1, \indent \varepsilon_2 = +1, \indent \varepsilon_3 = -1. 
	\end{equation}
	Similarly, one may calculate to verify that 
	\begin{eqnarray*} 
		\ket{U_{\l v}} &=& \varepsilon_\l~\factor{uv}~ ~\ket{V_{\l v}},
	\end{eqnarray*}
	for the same combination of $\varepsilon_\l$ as in Eq.~\ref{eq:epsa_val}.
	This completes the proof of the theorem.  
\end{proof}

Further, we take the opportunity to establish a direct relation between the sign parameter $\varepsilon_\l$ and the signs of eigenvalues $\gamma_\l$. 
\begin{lemma} \label{rel:epsa_l:egnval} For optimal $G$,
	\begin{equation} \varepsilon_\l = \sgn~\gamma_\l. \end{equation}
\end{lemma}    
\begin{proof} For optimal $G$, $\Gamma$ is a diagonal matrix with diagonal entries $\gamma_\l$. Thus, for signals sent in $\xy$ basis,
	\begin{eqnarray*} 
%		\gamma_\l &=& \tr\left(\Gamma_{xy} E_\l \right) = \tr\left(\rho_xE_\l\right) - \tr\left(\rho_yE_\l\right) = P_{\l x}-P_{\l y}. 
		\gamma_\l \!=\! \tr(\Gamma_{xy} E_\l) \!=\! \half[\tr(\rho_xE_\l) \!-\! \tr(\rho_yE_\l)] \!=\! \half(P_{\l x} \!-\! P_{\l y}). 
%		\gamma_\l &=& \tr\left(\Gamma_{xy} E_\l \right) = \half\left[\tr\left(\rho_xE_\l\right) - \tr\left(\rho_yE_\l\right)\right] \\
%		&=& \half\left(P_{\l x}-P_{\l y}\right).
	\end{eqnarray*}
	By Eq.~\eqref{eq:sign:el},
	\begin{eqnarray*} 
		\varepsilon_\l &=& \sgn\left(Q_{x\l} - Q_{y\l}\right) = \sgn\left(P_{\l x} - P_{\l y}\right) = \sgn~\gamma_\l,
	\end{eqnarray*}
	which establishes the relation.
\end{proof}
\begin{remark}
	By Lemma~\ref{rel:epsa_l:egnval}, another indication for optimality is that Eq.~\eqref{eq:epsa_val} should match with the signs of the eigenvalues $\gamma_\l$ of $\Gamma_{xy}$ as in Eq.~\eqref{eq:Eigenval}, which indeed happens here. Therefore, for any rotation of the canonical basis, a pair $\left(\matd\rot\vcan, ~\rot\vcan\right)$ of optimal interaction and corresponding optimal POVM in Eqs.~\eqref{eq:xi-zeta-I-final},~\eqref{eq:optimal_POVM-vec} achieves optimality.
%	Therefore, the optimality is achieved for the interaction given by Eq.~\eqref{eq:xi-zeta-I-final} and the class of POVMs corresponding to the eigenbasis given by Eq.~\eqref{eq:optimal_POVM-vec}.  
\end{remark}

Here we summarize the results achieved so far. While Eq.~\eqref{eq:xi-zeta-I-final} captures a class of optimal interactions in a unique form, Eq.~\eqref{eq:optimal_POVM-vec} depicts the class of optimal POVMs for those interactions. Combining Eqs.~\mref{eq:xi-zeta-I-final,eq:optimal_POVM-vec} we get the whole class of optimal interactions expressed in canonical basis. Fixing a rotation matrix $\rot$ then produces a particular instance of an optimal interaction, while varying $\rot$ produces the whole class of optimal interactions. Although there are infinitely many optimal interactions while expressed in canonical basis, they all have a unique form~\eqref{eq:xi-zeta-I-final} while written in eigenbasis of $\Gamma_{xy}$. Figure~\ref{fig:many2unq} illustrates this fact. Each round node in the figure denotes that the two input results together derive the output result.
\section{A Discussion on Connection between Our Results and those of Fuchs \etal. \cite{fuchs97}}
\label{connectfuchs}
Here we show that the instances of an optimal interaction presented in \cite{fuchs97} is a special case of the generalized unique form of the optimal interaction that we have derived. Moreover, we generate a new instance (different from the two instances of~\cite{fuchs97}) of the optimal interaction to add more clarity to our achievement.\smallskip 

\renewcommand{\arraystretch}{1.5}
\begin{sidewaystable}[htbp]
	\centering
	{\scriptsize
		\begin{tabular}{|p{1.4 cm}||p{2.9 cm}|p{2.8 cm}|p{3.7 cm}|p{7.7 cm}|p{3.7 cm}|}
			\hline \hline
			Optimal Interaction & Ours: General [Eq.~\ref{eq:xi-zeta-I-final}] & Fuchs 1 [Eq.~\ref{eq:opt_xi-zeta}] & Fuchs 2 [Eq.~\ref{eq:choice:xi_zeta:opt}] & Ours: One Parameter [Eq.~\ref{eq:xi-zeta-opt-gen}] & Ours: Special Case [Eq.~\ref{eq:xi-zeta-opt-gen-example1}]   \\
			\hline \hline				
			$\ket{\xi_x}$	& $\intDA$ & $\intAA$ & $\intBA$  & $\intEA$ & $\intCA$ \\
			$\ket{\xi_y}$	& $\intDB$ & $\intAB$ & $\intBB$  & $\intEB$ & $\intCB$ \\
			$\ket{\zeta_x}$ & $\intDC$ & $\intAC$ & $\intBC$  & $\intEC$ & $\intCC$ \\
			$\ket{\zeta_y}$ & $\intDD$ & $\intAD$ & $\intBD$  & $\intED$ & $\intCD$ \\    
			\hline \hline
			Optimal POVM & Ours: General & Fuchs 1 [Eq.~\ref{eq:optEl}] & Fuchs 2 [Eq.~\ref{eq:optimal_POVM-vec:eqErr}] & Ours: One Parameter [Eq.~\ref{eq:optimal_POVM-vec_exmpl}] & Ours: Special Case [Eq.~\ref{eq:optimal_POVM-vec-example1}]  \\
			\hline \hline			
			$\ket{E_0}$   & $\ket{E_0}$ & $\povmAA$ & $\povmBA$   & $\povmEA$ & $\povmCA$ \\
			$\ket{E_1}$   & $\ket{E_1}$ & $\povmAB$ & $\povmBB$   & $\povmEB$ & $\povmCB$ \\
			$\ket{E_2}$   & $\ket{E_2}$ & $\povmAC$ & $\povmBC$   & $\povmEC$ & $\povmCC$ \\
			$\ket{E_3}$   & $\ket{E_3}$ & $\povmAD$ & $\povmBD$   & $\povmED$ & $\povmCD$ \\ 
			\hline 
			\hline 
		\end{tabular}
	}
	\caption{Optimal interaction and corresponding optimal POVM - unique general form and its specific instances.}
	\label{tab:optintall}
\end{sidewaystable}

As discussed earlier, Eqs.~\eqref{eq:xi-zeta-I-final} and~\eqref{eq:optimal_POVM-vec} are key ingredients to generate different instances of the pair $\left(\matd\rot\vcan, ~\rot\vcan\right)$ of optimal interaction and corresponding optimal POVM by varying rotation $\rot$ of the canonical basis. 
For a special type of the orthogonal matrix $\rot$ as given by Example~\ref{ex1}, we combine these results and write an optimal interaction in terms of the canonical basis as below:
\begin{eqnarray} \label{eq:xi-zeta-opt-gen}
\ket{\xi_x} &=& \left(\Dp_{uv}\sqrt{a}+\Dm_{uv}\sqrt{1-a}\right)\ket{\E_0}, \nonumber\\
\ && + \left(\Dm_{uv}\sqrt{a}-\Dp_{uv}\sqrt{1-a}\right)\ket{\E_1}, \nonumber\\
\ket{\xi_y} &=& \left(\Dm_{uv}\sqrt{a}+\Dp_{uv}\sqrt{1-a}\right)\ket{\E_0} \nonumber\\
\ && + \left(\Dp_{uv}\sqrt{a}-\Dm_{uv}\sqrt{1-a}\right)\ket{\E_1}, \nonumber\\
\ket{\zeta_x} &=& \left(\Dp_{uv}\sqrt{a}+\Dm_{uv}\sqrt{1-a}\right)\ket{\E_2} \nonumber\\
\ && + \left(\Dm_{uv}\sqrt{a}-\Dp_{uv}\sqrt{1-a}\right)\ket{\E_3}, \nonumber\\
\ket{\zeta_y} &=& \left(\Dm_{uv}\sqrt{a}+\Dp_{uv}\sqrt{1-a}\right)\ket{\E_2} \nonumber\\
\ && + \left(\Dp_{uv}\sqrt{a}-\Dm_{uv}\sqrt{1-a}\right)\ket{\E_3}.  %\nonumber \\
\end{eqnarray}

\begin{figure*}[htbp]
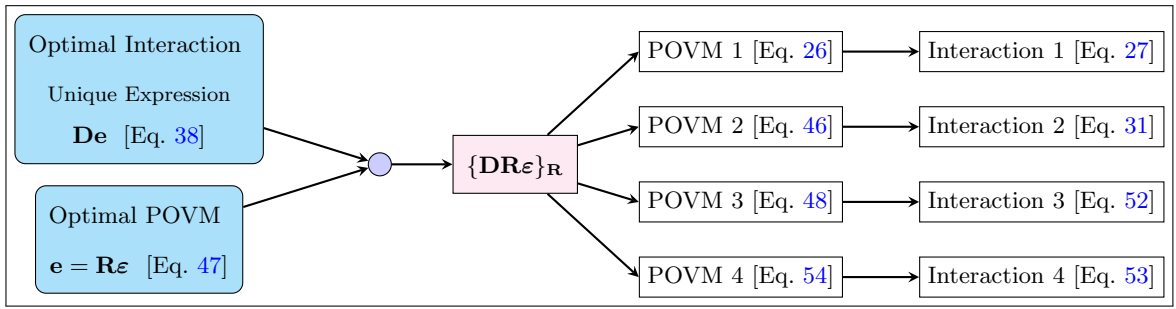

	\caption{Optimal interaction: unique expression to specific instantiations} \label{fig:unq2many}
	\centering \boxed{\tikzUnq}
\end{figure*}

For unequal error rates, Eq.~\eqref{eq:opt_xi-zeta} is a special case (apart from a permutation of the canonical basis) with $a=1$ in Eq.~\eqref{eq:xi-zeta-opt-gen}. Similarly, for equal error rates($D_{xy}=D_{uv}=D$), Eq.~\eqref{eq:choice:xi_zeta:opt} is a special case with $a=\Dp^2$ in Eq.~\eqref{eq:xi-zeta-opt-gen}. One may consider innumerable such optimal interactions (and corresponding optimal POVMs) by tuning the rotation parameter $a$ in the range $[0, 1]$. One such example is given below for unequal error rates.
\begin{example}
	Let $a=\half$. Thus the optimal interaction in Eq.~\eqref{eq:xi-zeta-opt-gen} becomes
	\begin{eqnarray} \label{eq:xi-zeta-opt-gen-example1}
	\ket{\xi_x} &=& \sqrt{1-D_{uv}}~\ket{\E_0} - \sqrt{D_{uv}}~\ket{\E_1} \nonumber,\\
	\ket{\xi_y} &=& \sqrt{1-D_{uv}}~\ket{\E_0} + \sqrt{D_{uv}}~\ket{\E_1} \nonumber,\\
	\ket{\zeta_x} &=& \sqrt{1-D_{uv}}~\ket{\E_2} - \sqrt{D_{uv}}~\ket{\E_3} \nonumber,\\
	\ket{\zeta_y} &=& \sqrt{1-D_{uv}}~\ket{\E_2} + \sqrt{D_{uv}}~\ket{\E_3}, 
	\end{eqnarray}
	and the corresponding optimal POVM is captured by
	\begin{eqnarray} \label{eq:optimal_POVM-vec-example1}
		\ket{E_0} = \invsqrttwo\left(\ket{\E_0}-\ket{\E_1}\right), & \ & 
		\ket{E_1} = \invsqrttwo\left(\ket{\E_0}+\ket{\E_1}\right), \nonumber\\
		\ket{E_2} = \invsqrttwo\left(\ket{\E_2}-\ket{\E_3}\right), & \ & 
		\ket{E_3} = \invsqrttwo\left(\ket{\E_2}+\ket{\E_3}\right). \nonumber \\
	\end{eqnarray}
\end{example}

One may easily check that the interaction presented here is indeed optimal. 
Clearly, the general form of the optimal interaction provided in this paper yields different choices of those in \cite{fuchs97}. Moreover, it's implementation is independent of equal or unequal error rates.

At this stage, we weigh the results achieved by us. Fuchs \etal.~\cite{fuchs97} came up with two different configurations for optimal interactions expressed in canonical basis. For the first configuration (Eq.~\ref{eq:optimal_xi-zeta}), they described the corresponding POVM (Eq.~\ref{eq:optEl}) w.r.t. the canonical basis, while for their second configuration (Eq.~\ref{eq:choice:xi_zeta:opt}), we have deduced the corresponding POVM (Eq.~\ref{eq:optimal_POVM-vec:eqErr}) in terms of the canonical basis. We have presented one more instance of an optimal interaction (Eq.~\ref{eq:xi-zeta-opt-gen-example1}) and the corresponding POVM (Eq.~\ref{eq:optimal_POVM-vec-example1}) w.r.t. the canonical basis. 
Table~\ref{tab:optintall} describes the general form of the optimal interaction and also shows its four specific instantiations, of which the first two coincide with those of Fuchs et al.~\cite{fuchs97} and the later two with our examples discussed earlier, the corresponding POVMs are also captured there.

For each of these three instances of the optimal interaction, one may use the relation between the eigenbasis and the canonical basis (looking at the POVM) to express the interaction w.r.t. the eigenbasis and notice that the final form becomes the same (Eq.~\ref{eq:xi-zeta-I-final}) for all these cases. It turns out that every possible instances of an optimal interaction written w.r.t. the canonical basis can be transformed to a unique description (Eq.~\ref{eq:xi-zeta-I-final}) in terms of the eigenbasis via the corresponding POVM. This is the significance of our work. We could establish that there exists infinitely many possible instances of an optimal interaction represented in a canonical basis, but they all have a unique representation while expressed in the eigenbasis. Feeding an optimal POVM to the unique form of the optimal interaction produces a specific instance of an optimal interaction. This is depicted in Figure~\ref{fig:unq2many}. 
Since an optimal interaction has an unique form and the form in Fuchs \etal.~\cite{fuchs97} is a special case, any instance of such an optimal interaction will achieve the same optimal information gain $G^{\star}$ benchmarked in~\cite{fuchs97}, neither more nor less. \medskip
%Thus, we do not claim any improvement on the optimal information gain here. 
  
\section{Conclusion}
For the BB84 quantum key distribution protocol, we have established a unique form describing the optimal interaction followed by the class of corresponding optimal measurements for the optimal information gain an eavesdropper can obtain for a given average disturbance when her interaction and measurements are performed signal by signal. We have shown that the choice of optimal interaction in~\cite{fuchs97}, for equal as well as unequal error rates, is a special case of the form provided by us.

\end{document}